%% file: arxiv_oct25.tex
\documentclass[dvipsnames,letterpaper,11pt]{article}
\usepackage[margin=1in]{geometry}
\usepackage{booktabs} 
\usepackage[ruled]{algorithm2e} 
\usepackage{xfp} 
\usepackage{tikz}
\usepackage{tikz-network}
\usetikzlibrary{patterns,positioning}
\usetikzlibrary{arrows}
\usetikzlibrary{arrows.meta, decorations.markings}
\usetikzlibrary{patterns.meta}
\usepackage[hidelinks]{hyperref}
\usepackage{cleveref}
\usepackage{amsthm}
\usepackage{stmaryrd}
\usepackage[longnamesfirst,sort]{natbib}
\usepackage{caption}
\usepackage{subcaption}
\usepackage{float}
\usepackage{comment}
\usepackage{cancel}
\usepackage{xcolor}
\usepackage[T1]{fontenc}
\usepackage{todonotes}

\definecolor{blue}{HTML}{34495E}
\definecolor{red}{HTML}{C0392B}
\definecolor{green}{HTML}{F1C40F}
\definecolor{grey}{HTML}{BDC3C7}

\usepackage{amsmath}
\usepackage{tikz}
\usepackage{pgfplots}
\usepackage{pgfplotstable}
\usepackage{soul}
\usepackage{mathtools}

\usepgfplotslibrary{external}
\pgfplotsset{compat=1.3}

\usepackage{thmtools} 
\usepackage{thm-restate}

\SetAlFnt{\small}
\SetAlCapFnt{\small}
\SetAlCapNameFnt{\small}
\SetAlCapHSkip{0pt}
\IncMargin{-\parindent}
\usepackage{enumitem}
\usepackage{bbold}

\newcommand{\NN}{\mathbb{N}}

\newcommand{\PP}{\mathbb{P}}
\newcommand{\RR}{\mathbb{R}}
\newcommand{\EE}{\mathbb{E}}

\newcommand{\calA}{\mathcal{A}}

\newcommand{\Gr}{\mathrm{Gr}}
\newcommand{\spl}{\hat{v}}
\newcommand{\boost}{\hat{V}}

\newtheorem{theorem}{Theorem}

\newtheorem{corollary}{Corollary}

\newtheorem{claim}{Claim}

\crefname{claim}{Claim}{Claims}

\newcommand{\email}[1]{\protect\href{mailto:#1}{#1}}

\begin{document}
	
\title{Online Proportional Apportionment\\
}
\author{Javier Cembrano\thanks{Department of Algorithms and Complexity, Max Planck Institut für Informatik; Department of Industrial Engineering, Universidad de Chile; \email{jcembran@mpi-inf.mpg.de}.}
\and Jose Correa\thanks{Department of Industrial Engineering, Universidad de Chile; \email{correa@uchile.cl}.}
\and Svenja M.\ Griesbach \thanks{Centro de Modelamiento Matemático (CNRS IRL2807), Universidad de Chile; Department of Computer Science, RWTH Aachen University; \email{sgriesbach@cmm.uchile.cl}.}
\and Victor Verdugo\thanks{Institute for Mathematical and Computational Engineering and Department of Industrial and Systems Engineering, PUC Chile; \email{victor.verdugo@uc.cl}.}
}

\date{}
\maketitle

\begin{abstract}
Traditionally, the problem of apportioning the seats of a legislative body has been viewed as a one-shot process with no dynamic considerations. While this approach is reasonable for some instances of the problem, dynamic aspects play an important role in many others. In this paper, we initiate the study of apportionment problems in an online setting. Specifically, we introduce an online algorithmic framework to handle proportional apportionment with no information about future events. 
In this model, time is discrete and there are~$n$ parties that receive a certain share of the votes at each time step.
An online algorithm needs to irrevocably assign a prescribed number of seats at each time, ensuring that each party receives its fractional share rounded up or down, and that the cumulative number of seats allocated to each party remains close to its cumulative share up to that time.

We consider deterministic and randomized online apportionment methods. For deterministic methods, we construct a family of adversarial instances that yield a lower bound, linear in $n$, on the worst-case deviation between the seats allocated to a party and its cumulative share.
We show that this bound is best possible and is matched by a natural greedy method.
As a consequence, a method guaranteeing that the cumulative number of seats assigned to each party up to any step equals its cumulative share rounded up or down (global quota) exists if and only if~$n\leq 3$.
Then, we turn to randomized allocations and show that, when $n\leq 3$, we can randomize over methods satisfying global quota with the additional guarantee that each party receives, in expectation, its proportional share in every step.
Our proof is constructive: We show that any method satisfying these properties can be obtained from a flow on a recursively constructed network.
We showcase the applicability of our results to obtain approximate solutions in the context of online dependent rounding procedures for multidimensional instances.
\end{abstract}

\thispagestyle{empty}
\newpage
\setcounter{page}{1}

\section{Introduction}

The {\it proportional apportionment} problem plays a paramount role in the structural electoral design of modern democratic systems, capturing the task of allocating the seats of a representative body (e.g., a parliament or constitutional assembly).
Formally, an instance is given by a positive integer vector $v=(v_1,\ldots,v_n)$ whose entries sum up to a positive integer value $H$.
The goal is to find a non-negative integer vector $a=(a_1,\ldots,a_n)$ such that $\sum_{i=1}^na_i=H$, and $a_i$ should be {\it proportional} to $v_i$.
This problem is typically solved in the following electoral scenarios.
In one of them, we have a set of $n$ states or districts, and each value $v_i$ represents the fraction of the population of state $i$.
In a second scenario, we have a set of $n$ political parties, and each value $v_i$ represents the fraction of votes obtained by party $i$.

Historically, the quality of an apportionment method has been measured according to whether it satisfies a prescribed set of axioms or not. 
Among them, the {\it quota property}, requiring that each party gets its quota rounded up or down, has played a prominent role since Alexander Hamilton, the first US Secretary of the Treasury, proposed a solution to the apportionment problem in 1791, later adopted in 1851.
Hamilton's method first calculates each party's quota, immediately assigns the floor of this number to the party, and then it goes through the parties in decreasing order of their quota residue $v_i - \lfloor v_i \rfloor$ and assigns one more seat to each party until the desired house size $H$ is met; by construction, Hamilton's method satisfies the quota property.
We refer to the books by \cite{balinski2010fair} and \cite{pukelsheim2017} for an extensive mathematical treatment of the quota property and the axiomatic theory of proportional apportionment.

While the classic apportionment theory provides a thorough axiomatic characterization of proportionality, it takes a static point of view in which solutions for an apportionment instance are computed independently of the past and future events.
Under this paradigm, when an apportionment problem is solved repeatedly over time, rounding decisions made {\it locally} at each time do not prevent systematic biases over the cumulative solutions.
However, dynamic representation plays a crucial role in the political organization of societies and the implementation of public policies~\citep{stimson1995dynamic}.
In addition to the natural application to elections that are repeated after each term of office---as discussed by \citet{golz2025apportionment} in the context of the sortition of the European Commission---the challenge of proportional apportionment over time arises in contexts beyond political representation. 
For instance, the Indian reservation policy mandates that publicly funded institutions set aside a fixed share of seats and jobs for designated beneficiary groups, and entitlements that must be rounded to whole numbers~\citep{evren2021affirmative}.
As institutions run successive recruitment rounds, they must ensure their total integral apportionment allocations stay as close as possible to the prescribed cumulative (fractional) quotas; this {\it global} goal is in tension with the allocation monotonicity of reserved seats over the recruitment cycles.
Similar problems are faced by institutions in other contexts, such as the distribution of human resources, public services, and facilities across different organizational or geographical units.

The key problem behind the tension between the local and global quota requirements is that, although we have access to historical data, we lack information about future apportionment instances; a decision taken today can induce deviations from the cumulative quotas for the next period.
Consider the following example (introduced by \citet{golz2025apportionment}) with sequential elections to allocate a single seat among four parties $\{1,2,3,4\}$.
In the first step, each party receives $1/4$ of the votes, so we can assume that a method assigns the seat to an arbitrary party, say party $1$.
In the second step, parties $2$, $3$, and $4$ get $1/3$ of the votes each; we can again assign the seat arbitrarily, say to party $2$.
If in the third step the votes are evenly distributed among parties $3$ and $4$, the one not receiving the seat will have a cumulative entitlement strictly larger than $1$ but no seat allocated.
Therefore, the global quota requirement is violated regardless of the choices made in the method.

\subsection{Our Contribution}

We introduce a modeling and online-algorithmic framework to handle proportional apportionment problems under local and global rounding considerations with no information about future events.
In the following, we summarize the contributions of our work and discuss its consequences beyond apportionment.

\paragraph{A model for online proportional apportionment.} 
We propose a new model for online apportionment with proportional representation.
In this model, there are~$n\in\NN$ parties, and at each time step~$t\in\NN$, they receive a fractional allocation of votes represented by a vector~$v^t$ summing to an integer~$H^t\in\NN$. We consider methods satisfying local quota, meaning that we allocate either $\lfloor v^t_i\rfloor$ or $\lceil v^t_i\rceil$ seats to each party $i\in [n]$. Therefore, throughout the paper we assume without loss of generality that $v^t\in [0,1)^n$.
Upon observing~$v^t$, a method needs to irrevocably assign exactly~$H^t$ seats to parties with strictly positive votes, granting one seat per selected party.
Crucially, the decision at time~$t$ can depend only on observed votes and allocations up to that time.
Future votes cannot be observed or influence the decision.
Conversely, the votes can be set based on the previous choices of the method.

We study which fairness properties can be achieved in this online setting. 
In particular, we focus on three central notions that apply to each party and at each step:
global quota (the total number of allocated seats matches the floor or ceiling of the total number of votes),~$\alpha$-proportionality
(the total number of seats and the total number of votes differ by at most~$\alpha$),
and ex-ante proportionality (the expected total number of seats equals the total number of votes).
While in the offline setting, where the entire sequence of votes is known in advance, there exist allocation methods that satisfy global quota (and therefore~$1$-proportionality) and ex-ante proportionality for any number of parties (\Cref{prop:offline}), we show that achieving similar guarantees in the online setting is significantly more challenging.

\paragraph{Deterministic methods and approximate proportionality.} 
In \Cref{sec:deterministic}, we study deterministic allocation methods and introduce the greedy apportionment method, which in each step assigns the~$H^t$ seats to the parties whose current number of allocated seats deviates most from their cumulative share.
As the first part of our main result (\Cref{thm:app-quota}) we show that this method is~$\smash{\frac{n-1}{2}}$-proportional for any number of parties~$n\in\NN$, and even satisfies global quota when~$n\leq 3$.
The analysis builds on the closely related \emph{leaky-bucket problem}, where \cite{adler2003proportionate} proved~$(H_n-1)$-proportionality (with~$H_n$ denoting the~$n$th harmonic number).
However, the fact that our setting only allows allocations to parties with positive votes worsens the attainable guarantee from logarithmic to linear.
For the second part of \Cref{thm:app-quota}, we show that this bound is tight by constructing a family of hard instances that, for any~$\varepsilon>0$, rule out the possibility of achieving better than~$\big(\smash{\frac{n-1}{2}-\varepsilon}\big)$-proportionality with any deterministic method.

\paragraph{Randomized methods and ex-ante proportionality.} 
In \Cref{sec:online}, we turn to randomized methods.
The main result of this section (\Cref{thm:exante}) is that an online apportionment method satisfying both global quota and ex-ante proportionality exists if and only if~$n\leq 3$.
Our proof is constructive and based on a class of methods derived from a carefully designed flow network for each time step, which encodes all constraints using the full history of votes and allocations.
Moreover, we show that any method satisfying global quota and ex-ante proportionality can be obtained by this recursive flow construction, thereby fully characterizing the class of such methods (\Cref{prop:methods-one-to-one}).
We conclude by briefly explaining why this construction breaks down when~$n\geq 4$.

\paragraph{Consequences in online dependent rounding.}
Our results in Sections \ref{sec:deterministic} and \ref{sec:online} translate directly to the setting of online dependent rounding for multidimensional instances, i.e., we receive an $n$-dimensional marginal vector at each step instead of a single non-negative real value. 
In particular, the case of $n=1$ party captures the online level-set problem studied very recently by \cite{naor2025online}. 
In \Cref{sec:connections}, we showcase the applicability of \Cref{thm:app-quota} and \Cref{thm:exante} in the context of online optimization, and introduce a multidimensional generalization of the multi-stage stochastic covering problem studied by \cite{naor2025online}.
We describe how to achieve online near-feasibility and approximation guarantees as a simple consequence of our results. 

\subsection{Further Related Work}

Not only from a political viewpoint apportionment has been studied for centuries, but also its formal mathematical treatment has a long history, as testified by the books 
of \citet{balinski2010fair} and \cite{pukelsheim2017}.
Of particular relevance to our work is the \textit{quota} axiom, analyzed in depth by \citet{balinski1975quota}.

Our work departs from classic apportionment theory in two key ways: We consider \textit{online} methods that allocate seats in sequential elections without information about future elections, and part of our results concern \textit{randomized} apportionment.
While we initiate the study of online apportionment, \cite{evren2021affirmative} previously explored (offline) methods ensuring proportionality in sequential allocation steps, motivated by public sector recruitment in India.
Their positive results on monotone, proportional-in-expectation, and quota-approximating methods build upon the work of \citet{akbarpour2020approximate} on the decomposition of proportional fractional allocations into proportional-in-expectation and near-feasible integral allocations.
On the other hand, randomized apportionment has received increasing attention recently, decades after the foundational work of \citet{grimmett2004stochastic}.
Indeed, \citet{golz2025apportionment} established the existence of randomized methods satisfying quota, house monotonicity, and expected proportionality; \citet{cembrano2025new} studied randomized methods that allow minor deviations from the house size to circumvent impossibility results; and \citet{correa2024monotone} focused on monotonicity with respect to subsets of parties.
\citet{aziz2019random} also considered lotteries over deterministic apportionment methods in the context of strategyproof peer selection.

From a technical point of view, our work builds on a line of work applying optimization techniques to apportionment.
In particular, network flows have been exploited in the study of bidimensional proportionality~\citep{balinskidemange1989a,balinskidemange1989b,gaffkepukelsheim2008,guenterzachariasen2007} and house-monotone apportionment~\citep{cembrano2025new}; \citet{pukelsheimricca2011} provide an overview of these and other applications.
\citet{golz2025apportionment} also used linear programming techniques in the context of monotone apportionment. 
Also, our positive result for $n\leq 3$ generalizes the expected proportionality and quota properties of a rounding scheme proposed by \citet{naor2025online} as an online version of the classic scheme by \citet{srinivasan2001distributions}.

Also related to our setting is the \emph{bamboo garden trimming} problem introduced by \citet{gkasieniec2017bamboo}.
In this problem, each bamboo grows with a fixed rate (analogous to parties receiving the same number of votes in each election). In each step, one bamboo can be cut to the ground, and the objective is to minimize the maximum height across all bamboos over time.
The best-known approximation guarantees are due to \citet{kawamura2024proof}.
A generalization of this setting is captured by the \emph{cup game}, often studied in the context of load balancing~\citep{liu1973scheduling} and deamortization~\citep{dietz1987two}.
Here, an adversary distributes~$H$ fractional units of water across a subset of~$n$ cups, and the player selects~$H$ cups from which to remove one unit of water each.
This models a scenario with multiple processors or allocations per round and has been analyzed against both oblivious and adaptive adversaries~\citep{bender2019achieving,bender2021randomized,kuszmaul2020achieving}. 
Our analysis of the deterministic greedy apportionment method is closely related to the \emph{leaky bucket problem} studied by \citet{adler2003proportionate}.
In this model, each bucket starts at the same water level and leaks a fractional amount of water in each step.
The goal is to assign integral units of water to buckets in a way that maximizes the minimum fill level over time.
This setting differs from ours in a crucial aspect:
Allocations can be made to any bucket, regardless of its leak rate (or, in our interpretation, even to parties with zero votes). 
A more detailed comparison to our model is provided in \Cref{sec:deterministic}.

\section{Instances, Axioms, and Offline Methods}\label{sec:prelims}

Let $\NN$ (resp. $\NN_0$) denote the strictly positive (resp.\ non-negative) integers.
Let $[n]\coloneqq \{1,2,\ldots,n\}$ (resp.\ $[n]_0\coloneqq \{0,1,\ldots,n\}$) denote the set containing the first elements of these sets up to $n\in \NN$ (resp.\ $n\in \NN_0$).
We identify $[n]$ with a set of $n$ parties.
    In each election~$t\in \NN$, party~$i\in [n]$ receives a fractional seat entitlement~$v^t_i\in [0,1)$ based on the outcome of a vote;\footnote{We recall that the restriction to $v^t_i\in [0,1)$ is without loss in the context of methods that satisfy local quota. This is because, given any vote vector~$v^t\in \mathbb{R}^n_{\geq 0}$, any such method deterministically assigns each party its lower quota~$\lfloor v^t_i\rfloor$, which reduces the problem to allocating the remaining fractional seats~$v^t_i-\lfloor v_i^t\rfloor\in[0,1)$ among the~$n$ parties.}~$H^t\coloneqq\sum_{i=1}^n v^t_i \in \NN$ denotes the total number of seats to be distributed in step~$t$.
We interpret~$v^t_i$ directly as the votes party~$i$ receives in election~$t$ and we write~$v^t\in [0,1)^n$ for the corresponding vector; an \textit{$n$-dimensional instance} is fully described by the sequence~$(v^t)_{t\in \NN}$.
While we define our general instances in the online setting as infinite sequences to capture the fact that we do not know when they will end or what the votes in the next steps will be, we restrict to a finite number of steps when we study offline methods, as well as for the iterative construction of some methods and negative instances.
We refer to a vote sequence over a finite number of steps $T$, $(v^t)_{t\in \{T_0,\ldots,T_0+T\}}$, as a \textit{partial $n$-dimensional instance}.

We now introduce the notion of an offline method and the two key axioms we study throughout this paper.
A \textit{deterministic offline apportionment method} receives a partial $n$-dimensional instance $(v^t)_{t\in [T]}$ and outputs a sequence $(X^t)_{t\in [T]}$, where for every $t\in [T]$, $X^t\subset \{i\in [n]: v^t_i>0\}$ and $|X^t|=H^t$.
In words, an offline apportionment method selects a subset of parties receiving the $H^t$ seats in each step, constrained to those parties receiving strictly positive votes.
When the vote sequence is fixed, we identify a method as the sequence $(X^t)_{t\in [T]}$ (and not as a function mapping a vote sequence to a set sequence).
A \textit{randomized offline apportionment method} is a lottery over deterministic offline apportionment methods.

For an instance $(v^t)_{t\in [T]}$ and $t\in [T]$, we let $V^t_i\coloneqq \sum_{k\in [t]} v^k_i$ denote the votes received by party $i$ up to step $t$.
For an instance $(v^t)_{t\in [T]}$, a (deterministic or randomized) method $(X^t)_{t\in [T]}$, and $t\in [T]$, we let $A^t_i\coloneqq |\{k\in [t]: i\in X^k\}|$ denote the (possibly random) number of seats received by party $i$ up to step $t$.
A deterministic method $(X^t)_{t\in [T]}$ satisfies \textit{global quota} on a certain instance if, for every $t\in [T]$ and $i\in [n]$, it holds $A^t_i \in \{\lfloor V^t_i\rfloor, \lceil V^t_i\rceil\}$.
A randomized method satisfies global quota if all deterministic methods in its support satisfy global quota.
In addition, a randomized method satisfies \textit{ex-ante proportionality} if $\PP[i\in X^t]=v^t_i$ for every $t\in [T]$ and $i\in [n]$.

The following proposition states the existence of offline apportionment methods satisfying both of the axioms introduced above.
It follows from the existence of house monotone and quota compliant methods in the classic (non-sequential) apportionment problem, shown via different approaches by \citet{balinski2010fair,cembrano2025new,golz2025apportionment}.

\begin{restatable}{proposition}{propOffline}\label{prop:offline}
    For every~$n\in \NN$ and $n$-dimensional instance $(v^t)_{t\in [T]}$, there exists an offline apportionment method that satisfies global quota and ex-ante proportionality.
\end{restatable}
The proof of this proposition, which is included for completeness and can be found in \Cref{app:prop-offline}, models apportionment methods via flows on a capacitated network, as illustrated in \Cref{fig:offlineflownetwork}.
This construction was given by \citet{cembrano2025new} to prove the existence of apportionment methods satisfying house-monotonicity and quota compliance in a non-sequential setting, where house monotonicity requires that an increase in the total number of seats to allocate does not cause a drop in the number of seats any party receives.
Specifically, we consider an origin $o$ and vertices $(u^t)_{t\in [T]}$, where there is an arc $(s,u^t)$ with capacity $H^t$ for each $t\in [T]$. 
For each $t\in [T]$, there are also vertices $(w^t_i)_{t\in [T],i\in [n]}$ and arcs $(u^t,w^t_i)$ with capacity $1$ for each $i\in [n]$ such that $v^t_i>0$.
There are also arcs $(w^t_i,w^{t+1}_i)$ for every $t\in [T-1]$ and $i\in [n]$, with lower capacity $\lfloor V^t_i\rfloor$ and upper capacity $\lceil V^t_i\rceil$.
Finally, there is a destination~$d$ and arcs $(w^T_i,d)$ for every $i\in [n]$, with lower capacity $\lfloor V^T_i\rfloor$ and upper capacity $\lceil V^T_i\rceil$.
The flow on each arc $(u^t,w^t_i)$ is interpreted as the number of seats ($0$ or $1$) that party $i\in [n]$ receives in step $t$, so that the flow on each arc $(w^t_i,w^{t+1}_i)$ (or $(w^T_i,d)$ when $t=T$) corresponds to the total number of seats that party $i$ has received up to step $t$.
Any (deterministic) apportionment method corresponding to a feasible flow on this network satisfies global quota because of these arcs' capacities.
The fact that there exists a lottery over such methods satisfying ex-ante proportionality follows from two facts: (i) the fractional flow that sends a flow of $v^t_i$ on each arc $(w^t_i,w^{t+1}_i)$ (or $(w^T_i,d)$ when $t=T$) is feasible, and (ii) the network flow polytope has integral extreme points.
Hence, the proportional fractional apportionment can be obtained as a distribution over deterministic apportionment methods satisfying global quota.
\input{offlineflownetwork}

\section{Deterministic Methods and Approximate Proportionality}\label{sec:deterministic}

In this section, we study \emph{deterministic} online apportionment methods.
As our main result in this context, we provide tight bounds on their worst-case deviations from exact proportionality, understood as the maximum difference, in any step, between the number of seats and the cumulative share that a party has received.
Our findings have immediate implications regarding the axioms introduced in \Cref{sec:prelims}: In contrast to the positive result for offline methods, an online method can satisfy global quota if and only if there are at most three parties.
A comparison to results on the closely related \textit{leaky bucket} problem highlights the difficulties of our setting due to the restriction of only assigning seats to parties receiving positive votes in each step.

To formally present our results, we first formally define a method in the online setting.
A (deterministic) \textit{online apportionment method} selects, given an~$n$-dimensional instance~$(v^t)_{t\in \NN}$, a subset of parties~$X^t\subset [n]$ for every $t\in \NN$ as follows:
\begin{enumerate}[label=\normalfont(\roman*)]
    \item For $t=0$, we define~$V^0_i=v^0_i\coloneqq 0$ and~$A^0_i=a^0_i\coloneqq 0$ for every~$i\in [n]$.
    \item For~$t\in \NN$, the \textit{history} is given by the pair~$\theta^{t-1}\coloneqq (V^{t-1},A^{t-1})\in \Theta^{t-1}$, where~$\Theta^{t-1} \coloneqq  \RR^n_+\times \NN^n_0$, and we define the cumulative vote vector as $V^{t-1}_i \coloneqq  \sum_{k=0}^{t-1} v^k_i$ for every $i\in [n]$, and the cumulative allocation vector as~$A^{t-1}_i\coloneqq \sum_{k=0}^{t-1}a^k_i$.
    The function~$X^t\colon \Theta^{t-1}\times [0,1)^{n} \to \binom{[n]}{H^t}$ maps the pair~$(\theta^{t-1},v^t)$ to a subset~$X^t(\theta^{t-1},v^t)\subset \{i\in [n]: v^t_i>0\}$ of size~$H^t$, and the \textit{allocation} $a^t\in \{0,1\}^n$ is derived from this set:~$a^t_i=1$ if and only if~$i\in X^t(\theta^{t-1},v^t)$.
\end{enumerate}    
When the history and vote vector are fixed, we omit these arguments and use~$X^t$ to directly refer to the set (instead of the function).

To quantify the deviations from the proportional number of seats that parties should receive, we introduce the notion of \emph{approximate proportionality}, which measures how close the number of assigned seats to each party is to its cumulative votes.
For~$\alpha\geq 0$ and an~$n$-dimensional instance~$(v^t)_{t\in \NN}$, an online apportionment method~$(X^t)_{t\in \NN}$ is~\textit{$\alpha$-proportional} if, for every~$t\in \NN$ and~$i\in [n]$, the cumulative allocation satisfies
$
   |A^t_i-V^t_i|
   \leq \alpha
$.
We further say that a method is \textit{strictly~$\alpha$-proportional} if the above inequality holds strictly for every~$t\in \NN$ and~$i\in [n]$, i.e., $|A^t_i-V^t_i|< \alpha$.
The following theorem is the main result of this section.
\begin{restatable}{theorem}{thmAppQuota}\label{thm:app-quota}
    Let~$n\in \NN$.
    There exists an online apportionment method that is~$\frac{n-1}{2}$-proportional on every~$n$-dimensional instance, and strictly~$1$-proportional on every~$3$-dimensional instance.
    Conversely, for every constant~$\varepsilon>0$, there exists an~$n$-dimensional instance on which no online apportionment method is~$\big(\frac{n-1}{2}-\varepsilon \big)$-proportional.
\end{restatable}

As a consequence of \Cref{thm:app-quota}, we note the following observation regarding the global quota axiom introduced in \Cref{sec:prelims}.

\begin{corollary}\label{cor:global-quota}
	There exists an online apportionment method satisfying global quota on every $n$-dimensional instance if and only if $n\leq 3$.
\end{corollary}

In the remainder of this section, we prove Theorem \ref{thm:app-quota}.
Before presenting our analysis, we describe the closely related \textit{leaky bucket} problem \citep{adler2003proportionate}, using our notation to highlight the parallels. 
The problem can be described as follows.
A set of $n$ buckets is initialized with a common water level that we normalize to $0$.
In each step~$t\in \NN$, a volume $v^t_i\geq 0$ of water leaks from each bucket $i$, such that the total leakage~$H\coloneqq \sum_{i\in [n]}v^t_i$ is an integer; note that we allow negative loads.
An online algorithm must decide how to pour back~$H$ unit-volume refills, with the goal of maximizing the minimum load across all buckets and all time steps.
\citeauthor{adler2003proportionate} showed that a natural greedy algorithm, that refills the currently emptiest bucket one by one at each step, achieves a worst-case deviation from the initial load of at most~$H_n-1$ (where $H_n$ is the $n$th harmonic number), and that this guarantee is best possible.

While the leaky bucket problem is closely related to our setting, there is one key distinction:\footnote{In fact, \citet{adler2003proportionate} consider the deviation from the minimum load \textit{after} the water leaks but \textit{before} the algorithm refills the buckets. However, as noted in their work, this formulation is equivalent to ours up to an additive shift of~$H$.
}
In the apportionment setting, we can only assign seats to parties with positive votes, which corresponds to only allowing refills to buckets that leaked in the current step. 
In Theorem \ref{thm:app-quota}, we show that, perhaps surprisingly, this restriction significantly increases the worst-case deviation: it changes from logarithmic to linear.
Our lower bound is tight and is matched by a simple greedy method, which we describe below. 

\paragraph{The greedy apportionment method.} For an~$n$-dimensional instance~$(v^t)_{t\in \NN}$, we define the \emph{surplus} of party~$i$ at step~$t$ as~$s^t_i \coloneqq A^t_i-V^t_i$, where a negative surplus indicates a \textit{deficit}.
The surplus can be defined recursively by setting~$s^0_i\coloneqq 0$ and $s^t_i \coloneqq s^{t-1}+(a^t_i-v^t_i)$ for every~$t\in \NN$.
For simplicity, we assume throughout this section that, in each step $t$, the parties are sorted from larger to smaller surplus, i.e., $s^t_1\geq s^t_2\geq \cdots \geq s^t_n$.
This assumption is without loss of generality, since we can rename parties after each step to maintain this invariant without affecting the surplus vector~$s^t$.

For the positive direction of the theorem, we consider the natural \textit{greedy apportionment method} $(\Gr^t)_{t\in \NN}$, defined as follows.
For an $n$-dimensional instance~$(v^t)_{t\in \NN}$ and associated house sizes $(H^t)_{t\in \NN}$, the method assigns~$H^t$ seats in step~$t$ to the parties with the largest deficit up to step~$t-1$ plus votes in step $t$; i.e.,
\[
    \Gr^t(V^{t-1},A^{t-1},v^t)\in \arg\min \left\{\sum_{i\in S} (s^{t-1}_i-v^t_i): S\subset \{i\in [n]: v_i>0\},\ |S|=H^t \right\}.
\]
Each party~$i$ in the selected set is then assigned a seat in step~$t$, i.e.,~$a^t_i=1$ for every $i\in \Gr^t(V^{t-1},A^{t-1},v^t)$ and~$a^t_i=0$ for every other~$i$.
The following lemma establishes that this method achieves the proportionality guarantees stated in \Cref{thm:app-quota}.

\begin{restatable}{lemma}{lemGreedy}\label{lem:greedy}
    The greedy apportionment method is $\frac{n-1}{2}$-proportional on every~$n$-dimensional instance for every~$n\in \NN$, and strictly $1$-proportional on every~$3$-dimensional instance.
\end{restatable}

For the proof of this result, which is deferred to \Cref{app:lem-greedy}, we adapt the arguments by \citet{adler2003proportionate}, who established a worst-case deviation of~$H_n-1$ for the greedy algorithm in the leaky bucket problem. 
We outline the key ideas used to bound the maximum surplus; the bound on the deficit follows from analogous reasoning.

For each step~$t\in \NN_0$ and index~$\ell\in [n]$, we define~$\gamma^t_\ell$ as the average of the~$\ell$ largest surpluses at time~$t$, plus~$\ell/2$.
Observe that (i)~$\gamma^t_n=n/2$ for every~$t$; and (ii)~at step~$0$, the sequence~$(\gamma^0_\ell)_{\ell \in [n]}$ is strictly increasing.
Now suppose, for the sake of contradiction, that at some step~$t$ the maximum surplus is strictly larger than~$(n-1)/2$.
Then~$\gamma^t_1>n/2$, implying that there must exist some~$i\in[n-1]$ for which~$\gamma^t_i > \gamma^t_{i+1}$.
Let~$t$ and~$i$ be the earliest time and smallest index for which this occurs.

We now derive a contradiction by analyzing the change in surplus from step~$t-1$ to~$t$.
On the one hand, the total surplus of the~$i$ parties with the largest surplus at step~$t$ must have increased from~$t-1$ to~$t$, implying that some of these parties received a seat in step~$t$, despite the presence of a party outside this set that also received positive votes.
This suggests that the~$i$th and~$(i+1)$th largest surpluses are not too far apart; otherwise, the greedy method would have allocated a seat to the party outside the set of the~$i$ parties with the largest surplus.
On the other hand, by the choice of~$i$ and the definition of~$\gamma^t_\ell$, the gap between the~$i$th and~$(i+1)$ the largest surpluses must be large enough to ensure that~$\gamma^t_{i-1}<\gamma^t_i$, leading to a contradiction.

To establish \emph{strict} approximate proportionality for~$n=3$, we examine the only borderline case that does not lead to a contradiction by the argument above, namely~$i=1$.
Here, we apply a refined averaging argument and show that both the surplus and the deficit of some party would simultaneously exceed~$1$, contradicting the choice of the greedy method in the previous step.
We note that such a contradiction does not arise for other values of~$n$; for instance, when~$n=2$, one can construct a simple instance in which both parties attain a deviation of exactly~$1/2$.
 
\paragraph{Construction of the hard instances.} The proof of the negative direction of \Cref{thm:app-quota} is considerably more involved than the \textit{leaky bucket} case and is based on a careful iterative construction.
We consider instances in which each vote vector corresponds to a two-party election, i.e., the vector~$v^t$ contains exactly two non-zero entries for every~$t\in \NN$. 
The lower bound construction relies on a key property of such instances:
Whenever two parties have surpluses that differ by less than one, a two-party election between them can be used to split their surpluses so that they differ by exactly one, while preserving their average surplus.
This simple yet powerful observation forms the basis of our argument.

\begin{restatable}{lemma}{lemSplitters}\label{lem:splitters}
    Let $n\in \NN$, let $(v^k)_{k\in [t]}$ be a partial $n$-dimensional instance, and let $V^t$ denote the aggregate votes up to step $t$.
    Let $(X^k)_{k\in \NN}$ be any online apportionment method and $A^t$ the corresponding cumulative allocation up to step~$t$.
    Then, for all parties $i,j\in [n]$ with $i\neq j$ such that $0\leq s^t_i-s^t_j<1$, there exists a vote vector $v^{t+1}\in [0,1)^n$ such that, after the allocation $a^{t+1}$ corresponding to $X^{t+1}(V^t,A^t,v^{t+1})$,
    \[
        \vert s^{t+1}_i-s^{t+1}_j\vert = 1,\quad \frac{s^{t+1}_i+s^{t+1}_j}{2} = \frac{s^t_i+s^t_j}{2}.
    \]
\end{restatable}

For the proof, which is deferred to \Cref{app:lem-splitters}, we construct a two-party election~$v^{t+1}$ satisfying~$s^t_i-v^{t+1}_i=s^t_j-v^{t+1}_j$.
This ensures that, regardless of the allocation method, one of the two parties will end up with a surplus exactly one unit larger than the other after step~$t+1$.
As a direct consequence, this implies that~$\frac{1}{2}$-proportionality is tight and cannot be improved (even by a non-constant value) for instances with~$n=2$ parties.

In what follows, whenever two parties~$i$ and~$j$ satisfy the conditions of \Cref{lem:splitters} for a given history and apportionment method, we denote the corresponding vote vector by~$\spl(i,j)$ and refer to it as the \textit{$(i,j)$-splitter}.

We now state a lemma that implies the lower bound in \Cref{thm:app-quota}, and in fact yields a more general result.
It states that, for any partial~$n$-dimensional instance up to some step~$r$, any method, and any subset of~$n'$ parties, there exist~$t$ vote vectors that can be added to the partial instance such that the surplus or deficit of some party in this subset is arbitrarily close to~$(n'-1)/2$.

\begin{restatable}{lemma}{lemPropLB}\label{lem:prop-lb}
    Let $n,n'\in \NN$ with $n'\leq n$ and $P\subseteq [n]$ with $|P|=n'$ be a subset of~$n'$ parties in~$[n]$.
    Let $r\in \NN_0$ and $\varepsilon>0$ be fixed, and let $(v^k)_{k\in [r]}$ be a partial $n$-dimensional instance.
    Let $(X^t)_{t\in \NN}$ be any online apportionment method and $A^r$ the corresponding cumulative allocation up to step~$r$. 
    Then, 
    there exists a finite value $t\in \NN$ with $t\geq r$ and a vote sequence $(v^k)_{k\in \{r+1,\ldots,t\}}$ such that, 
    after the allocation $a^t$ corresponding to $X^t(V^{t-1},A^{t-1},v^t)$, there exists a party $i\in P$ such that~$|s^{t}_i| \geq (n'-1)/2-\varepsilon$.
\end{restatable}

The proof of this lemma is deferred to \Cref{app:lem-prop-lb}, but we outline the main ideas here.
The result is proven by induction on~$n'$ in steps of size~$2$.
The base case~$n'=1$ is trivial.
For~$n'=2$, the claim follows from the application of a single splitter:
If the two parties do not already differ in surplus by at least~$1$, the splitter can be used to separate their surpluses accordingly.

For the inductive step, consider a set~$P$ of~$n'+2$ parties.
We apply the inductive hypothesis to a subset of~$n'$ parties in~$P$, excluding the one with the largest surplus and largest deficit.
Repeating this process at most three times, we reach a situation in which two parties in~$P$ have surplus (or deficit) arbitrarily close to~$(n'-1)/2$, say, at least~$(n'-1-\varepsilon)/2$.

For simplicity, assume that two parties in~$P$ have a surplus (and not a deficit) arbitrarily close to~$(n'-1)/2$.
Applying a splitter to these two parties increases the surplus of the one with the larger surplus and decreases that of the other.
In particular, after this splitter, one of the two must have a surplus of at least~$(n'-1-\varepsilon)/2+1/2$.

To continue increasing this value, another splitter is needed.
However, since the other party's surplus is now one unit smaller, a direct application is not feasible.
Instead, we apply the inductive hypothesis once again to a subset of~$n'$ parties in~$P$, again excluding the two with the largest and smallest surplus.
As a result, at least one party has a surplus or a deficit exceeding~$(n'-1-\varepsilon)/2$. 
In the former case, we are again in the position to apply a splitter to the two parties with the largest surpluses, which this time will leave a party with a surplus of at least~$(n'-1-\varepsilon)/2+3/4$.
Otherwise, we now have two parties with a deficit exceeding~$(n'-1-\varepsilon)/2$, and we apply a splitter between them.

By repeatedly applying the induction hypothesis and a subsequent splitter to the surplus or deficit end, we increase the maximum deviation.
In particular, the deviation increases by at least~$1/2^\ell$ after the~$\ell$th iteration of this procedure at the corresponding end.
After finitely many steps, this process yields a surplus or deficit of at least~$(n'+1)/2-\varepsilon$.
\Cref{fig:splitterpath} illustrates the repeated application of splitters to get a surplus or deficit approaching~$(n-1)/2$, for~$n\in \{3,4\}$.
\input{splitterpath}

\paragraph{Relation to standard apportionment methods.}
When the same vote vector is repeated in every time step, an online apportionment method naturally corresponds to a house-monotone (offline) method, and vice versa.
Interestingly, the greedy apportionment method introduced in this work does not coincide with any of the standard offline rules.
Although it might appear similar to the Hamilton method, as both allocate seats sequentially to the parties with the largest deficit, this equivalence does not hold.
For instance, for three parties with vote vector~$v^t=(2/3, 29/120, 11/120)$, our greedy method yields the allocation~$(4,2,1)$ after seven steps, whereas the Hamilton method with seven seats outputs~$(5,2,0)$.
This discrepancy is closely related to the Alabama paradox, as the Hamilton method may withdraw a seat from a party when the house size increases, while an online method cannot revoke past allocations.

The Balinski-Young quota method, in contrast, satisfies both house-monotonicity and quota compliance and could therefore be adapted to operate in an online fashion while respecting global quota at every step.
This is particularly noteworthy since our greedy method, although optimal in the worst case, does not exhibit this property in this special setting.
To see this, consider three large parties and two small ones, with repeated vote vector~$v^t=(49/150,49/150,49/150,1/100,1/100)$.
After~$41$ steps, the cumulative allocation is~$(13,13,13,1,1)$.
At time~$t=42$, the next seat is assigned to one of the large parties,
as their deficits exceed those of the small ones. 
However, at~$t=43$ the other two large parties both have remainders~$43\cdot49/150-13=157/150>1$, forcing both to receive an additional seat to maintain global quota.

\section{Randomized Methods and Ex-ante Proportionality}\label{sec:online}
We now shift our focus to randomized online apportionment methods.
We introduce the formal notion of such a method in what follows.
A \emph{randomized online apportionment method} selects, given an~$n$-dimensional instance~$(v^t)_{t\in \NN}$, a random subset of parties~$X^t\subset [n]$  for every~$t\in \NN$, as follows:
\begin{enumerate}[label=\normalfont(\roman*)]
    \item For~$t=0$, we define~$V^0_i\coloneqq 0$ and~$A^0_i\coloneqq a^0_i\coloneqq0$ for every~$i\in [n]$.
    \item For~$t\in \NN$, we define the cumulative vote vector as~$V^{t-1}_i \coloneqq \sum_{k=0}^{t-1} v^k_i$ for every~$i\in [n]$, where~$v^0\coloneqq 0$, and the cumulative allocation vector as~$A^{t-1}_i\coloneqq \sum_{k=0}^{t-1}a^k_i$.
    The set of \textit{feasible histories}~$\Theta^{t-1}$ is given by all histories~$(V^{t-1},A^{t-1})$ that occur with strictly positive probability under the method up to step~$t-1$.
    There is a function
    \begin{align*}
    y^t\colon \Theta^{t-1}\times [0,1)^n \times \binom{[n]}{H^t} \to [0,1],
\end{align*}
such that for each feasible history~$(V^{t-1},A^{t-1})\in \Theta^{t-1}$, ~$y^t(V^{t-1},A^{t-1},v^t,\cdot)$ is a probability distribution over subsets of~$[n]$ of size~$H^t$ and~$y^t(V^{t-1},A^{t-1},v^t,S)=0$ for every~$S$ such that there is a party~$i\in S$ with~$v_i^t=0$.
The outcome in step~$t$ is then given by a random subset~$X^t(V^{t-1},A^{t-1},v^t)$ distributed according to~$y^t(V^{t-1},A^{t-1},v^t,\cdot)$, and the allocation~$a^t=a^t(V^{t-1},A^{t-1},v^t) \in \{0,1\}^n$ is derived from this set:~$a^t_i=1$ if and only if~$i\in X^t(V^{t-1},A^{t-1},v^t)$.
\end{enumerate}

If the method~$(y^t)_{t\in\NN}$ is only defined up to a finite time~$T\in \NN$, we refer to it as a \emph{partial online apportionment method}.
Given a method~$(y^t)_{t\in\NN}$ and an instance~$(v^t)_{t\in\NN}$, let~$\calA^t\coloneqq \{A^t: (V^t,A^t)\in \Theta^t\}$ be the set of cumulative allocation vectors that occur with strictly positive probability after step~$t$.
Moreover, let~$\pi(y^t,V^t,A^t)$ denote the probability of reaching allocation~$A^t$ at step~$t$ under method~$(y^t)_{t\in\NN}$.
We initialize with~$\calA^0\coloneqq\{0\}$,~$V^0\coloneqq0$, and~$\pi(y^{0},V^{0},A^{0})\coloneqq1$.
Now, ex-ante proportionality translates to the randomized setting as follows.
    An online apportionment method~$(y^t)_{t\in\NN}$ is \emph{ex-ante proportional} if for every step~$t\in\NN$ and party~$i\in[n]$, the expected number of seats assigned to~$i$ is equal to its share, i.e.,
    \begin{equation}
    \sum_{A^{t-1}\in\calA^{t-1}} \pi(y^{t-1},V^{t-1},A^{t-1}) \cdot a^{t}_i(V^{t-1},A^{t-1},v^t) =v^{t}_i.\label{eq:EAPdef}
    \end{equation}
We call a partial online apportionment method~$(y^t)_{t\in[T]}$ ex-ante proportional if it satisfies Equation~\eqref{eq:EAPdef} for every step~$t\in[T]$.

\Cref{cor:global-quota} states that an online apportionment method can satisfy global quota for all~$n$-dimensional instances~$(v^t)_{t\in\NN}$ if and only if~$n\leq 3$.
The main result of this section establishes that, for $n\leq 3$, we can randomize over methods satisfying global quota while additionally fulfilling ex-ante proportionality, thus giving ex-ante and ex-post proportionality guarantees.

\begin{theorem}\label{thm:exante}
    There exists an online apportionment method satisfying global quota and ex-ante proportionality for all~$n$-dimensional instances if and only if~$n\leq 3$.
\end{theorem}

We remark that the impossibility for $n\ge 4$ follows directly from \Cref{cor:global-quota}.
On the other hand, when $n=2$, a simple construction suffices to prove the positive direction.
Indeed, we can apply the systematic sampling procedure (a.k.a., {\it Grimmett's apportionment method}~\citeyearpar{grimmett2004stochastic}) to party~$1$, defined as follows:
Sample a value~$\lambda\in[0,1)$ uniformly at random and assign a seat to party~$1$ in step~$t$ if and only if~$\lfloor V_1^{t-1} + \lambda \rfloor<\lfloor V_1^{t} + \lambda \rfloor$.
This method satisfies global quota and ex-ante proportionality.\footnote{We remark that this observation was also made by \cite{naor2025online} in the context of online dependent rounding; in Section \ref{sec:connections}, we further develop the consequences of our results for dependent rounding algorithms in online settings.}
To obtain an online apportionment method for two parties, we can simply assign the seat in step~$t$ to party~$2$ whenever party~$1$ does not receive it.
Since~$v_2^t=1-v_1^t$ in every step, ex-ante proportionality for party~$1$ directly implies ex-ante proportionality for party~$2$.
Moreover, global quota is also satisfied: Whenever one party must be allocated a seat by the global quota constraint, the other party must not be allocated a seat due to the complementary constraint.

For three parties, however, simple constructions no longer suffice due to weaker dependencies between the global quota constraints of different parties.
Consider, for example, an instance~$(v^t)_{t\in\NN}$ where the first vote vector is~$v^1=(0.6,0.3,0.1)$ and suppose that in a particular realization, the seat in step~$1$ was allocated to party~$1$.
Now let~$v^2=(0,0.8,0.2)$.
Since party~$3$ did not receive a seat in step~$1$ and~$v^2_3>0$, global quota implies no restriction on whether it may or may not receive a seat in step~$2$.
At the same time, ex-ante proportionality requires that party~$3$ is assigned a seat with probability~$0.2$.
However, allocating a seat to party~$3$ under the current allocation would violate global quota for party~$2$, which also has not received a seat in step~$1$ but is surpassing its previous upper share in step~$2$.
This illustrates that, for~$n=3$, ensuring both global quota and ex-ante proportionality requires tracking and conditioning on the full history of seat allocations.
Consequently, designing valid online apportionment methods becomes significantly more intricate than in the two-party case.

We overcome this challenge and prove \Cref{thm:exante} by introducing a family of randomized online apportionment methods that we refer to as \emph{network flow methods}.
These methods are defined recursively and rely on constructing a suitable flow network in each step.
Consider an instance~$(v^k)_{k\in \NN}$ and a partial online allocation method~$(y^k)_{k\in [t]}$ that satisfies global quota and ex-ante proportionality.
For any cumulative allocation~$A^t\in\calA^t$, define the set of parties that have received their upper share of seats after step~$t$ as
\[
u(A^t)\coloneqq \Big\{ i\in[n] : A^t_i=\Big\lceil \textstyle\sum_{k=1}^t v_i^k\Big\rceil  \Big\}.
\]
Let~$U^t\coloneqq \{u(A^t) : (V^t,A^t)\in\Theta^t\}$ be the set of all such sets observed with strictly positive probability up to step~$t$.
With a slight overload of notation, we define
\[
\pi(u)\coloneqq\sum_{A^t\in \calA^t:u(A^t)=u}\pi(y^t,V^t,A^t)
\]
as the probability that the upper-quota set~$u\in U^t$ occurs after step~$t$.
We now define the flow network~$\text{FN}(V,A,v,\pi)$, which constitutes the basis of our recursive construction.

\paragraph{Flow network construction.}
Let~$(v^k)_{k\in \NN}$ be an~$n$-dimensional instance and~$(y^k)_{k\in [t]}$ a partial online apportionment method that satisfies global quota and ex-ante proportionality up to step~$t\in\NN$.
    Define~$U\coloneqq \{u(A): (V,A)\in\Theta^t\}$.
    The corresponding flow network~$\normalfont\text{FN}(V,A,v,\pi)$ for step~$t+1$ is given by
    \begin{enumerate}[label=\normalfont(\alph*)]
        \item vertex set~$P\coloneqq \{o,d\}\cup U \cup [n]$,
        \item edge set~$E\coloneqq \{(o,u): u\in U\} \cup \{(u,i): u\in U, i\in [n]\} \cup \{(i,d): i\in [n]\}$, and
        \item upper and lower capacities $c,\ell\colon E\to \RR_+$ given by
        \begin{align*}
            c((o,u))&=\pi(u),\\
            c((u,i))&=
            \begin{cases}
                0 &\text{ if } i\in u \text{ and } \lceil V_i \rceil =\lceil V_i+v_i\rceil
                \text{ or }v_i=0\\
                {\pi(u)}/{\sum_{i=1}^n v_i} &\text{ else, }\\
            \end{cases}\\
            c((i,d))&={v_i}/{\textstyle \sum_{i=1}^n v_i},\\
            \ell((o,u))&=0,\\
            \ell((u,i))&=
            \begin{cases}
                {\pi(u)}/{\sum_{i=1}^n v_i} &\text{ if } i\notin u \text{ and } \lfloor V_i\rfloor +1=\lfloor V_i+v_i\rfloor,\\
                0 &\text{ else, }\\
            \end{cases}\\
            \ell((i,d))&=0.
        \end{align*}
    \end{enumerate}
We illustrate the construction of a network flow method by the following example.
Consider an instance with three parties, and suppose that in the first election the votes are given by~$v^1=(0.6,0.3,0.1)$, which yields~$H^1=1$.
In this case, the unique partial network flow method at~$t=1$ is defined by
\begin{align*}
    y^1(V^0,A^0,v^1,\{i\})=v_i^1 \quad \text{ for all } i\in [3],
\end{align*}
which implies the edge capacities~$c(o,\{i\})=\pi(\{i\})=v^1_i$ in the flow network~$\normalfont\text{FN}(V^1,A^1,v^2,\pi)$.

Now consider the second vote vector~$v^2=(0.5,0.2,0.3)$ with~$H^2=1$.
The global quota property enforces that party~$1$ must receive a seat in step~$2$ if it was not allocated one in round~$1$, as~$\lfloor V_1^1 \rfloor +1=\lfloor V_1^1 +v_1^2 \rfloor$.
Conversely, for each party~$i\neq 1$, the condition~$\lceil V_i^1 \rceil =\lceil V_i^1+v_i^2\rceil$ implies that they must not be allocated a seat in round~$2$ if they were already selected in round~$1$.
These constraints are encoded in the capacities of the edges~$e=(u,i)$ in the flow network.
The current vote vector~$v^2$ also determines the capacities~$c((i,d))=v_i^2$ for every~$i\in[3]$.

The complete flow network is depicted in \Cref{fig:flownetwork}.
Edges with zero upper capacity are omitted, and edges with equal lower and upper capacities are highlighted in red.
The network admits a unique feasible~$(o,d)$-flow~$f$ of value~$1$, given by
\begin{align*}
    f(e)=
    \begin{cases}
        0.1 & \text{ if } e=(\{1\},1),\\
        0.2 & \text{ if } e=(\{1\},2),\\
        0.3 & \text{ if } e=(\{1\},3),\\
        0 & \text{ if } e=(\{2\},3) \text{ or } e=(\{3\},2),\text{ and } \\
        c(e) & \text{ otherwise.}
    \end{cases}
\end{align*}
For a detailed analysis of feasibility, we refer to Case~1.1.2 in the proof of \Cref{lem:flowpolytope} in \Cref{app:lem:flowpolytope}.
\input{flownetwork}

We now describe how to construct a \emph{partial network flow method}~$(y^k)_{k\in [t]}$ recursively for any number of parties~$n$.

\paragraph{Step~$t=1$.}
We define~$y^1$ as a probability distribution over subsets of~$[n]$ of size~$H^1$.
Since vector~$v^1$ lies in the hypersimplex~$\Delta(n,H^1)$, there exists a convex decomposition~$v^1=\sum_{j=1}^m \lambda_j\cdot z_j$, where each~$z_j\in \{0,1\}^n$ has exactly~$H^1$ ones and~$\lambda_1,...,\lambda_m\in [0,1]$ with~$\sum_{j=1}^m\lambda_j=1$.
Let~$Z_j\subseteq [n]$ denote the support of~$z_j$, i.e.,~$i\in Z_j$ if and only if~$z_{j,i}=1$.
We then define 
\begin{align*}
    y^1(V^0,A^0,v^t,Z_j)&=\lambda_j &&\text{ for every }j\in[m] \text{ and }\\
    y^1(V^0,A^0,v^t,S)&=0 && \text{ for all other subsets } S\subseteq [n].
\end{align*}
\paragraph{Step~$t\ge 2$.}
Assume we have already constructed a partial network flow method~$(y^k)_{k\in [t]}$.
Let~$V^t$, and~$A^t$ denote the cumulative vote and allocation vectors, and~$\pi(u)$ the distribution over upper-quota sets~$u\in U^t$.
Construct the flow network~$\text{FN}(V^t,A^t,v^{t+1},\pi)$ for step~$t+1$.
If this network admits a feasible~$(o,d)$-flow~$f$ of value~$1$, we can extend the partial method to step~$t+1$ as follows.

Fix a feasible allocation~$A^t$ that occurs with strictly positive probability.
Let~$u=u(A^t)$ denote the corresponding upper-quota set.
Using flow~$f$, we define the fractional assignment vector~$z(u)\in[0,1]^n$ by
\[
z_i(u)\coloneqq f(u,i)\cdot \frac{H^{t+1}}{\pi(u)}.
\]
By construction,~$z(u)$ lies in the hypersimplex~$\Delta(n,H^{t+1})$.
Again, decompose~$z(u)=\sum_{j=1}^m \lambda_j \cdot z_j$ into a convex combination of~$\{0,1\}$-vectors~$z_j\in \{0,1\}^n$ with exactly~$H^{t+1}$ ones and let~$Z_j$ be their supports.
We define
\begin{align*}
    y^{t+1}(V^t,A^t,v^{t+1},Z_j) &= \lambda_j &&\text{ for every }j\in[m] \text{ and }\\
    y^{t+1}(V^t,A^t,v^{t+1},S)&=0 && \text{ for all other subsets } S\subseteq [n].
\end{align*}
This completes the construction of the partial network flow method.

We now show that a partial network flow method satisfies all three desired properties.

\begin{restatable}{lemma}{lemflownetworkproperties}\label{lem:flownetwork-properties}
    Let~$(y^t)_{t\in[T]}$ be a partial network flow method.
    Then, for every~$t\in [T]$ and every feasible history~$(V^t,A^t)\in\Theta^t$ that occurs with strictly positive probability, we have~$A^t_i\in\{\lfloor V_i^t\rfloor,\lceil V^t_i\rceil\}$
    for every~$i\in[n]$.
    Furthermore,~$(y^t)_{t\in[T]}$ satisfies Equation~\eqref{eq:EAPdef} for every~$t\in [T]$.
\end{restatable}

The proof uses induction to show that a partial network flow method satisfies global quota and ex-ante proportionality by construction.
We verify that only parties with strictly positive votes are assigned seats and that these lie within the share bounds using the structure of the network.
The flow conservation finally yields that the expected number of seats assigned in each step matches the vote vector.
The full proof is deferred to \Cref{app:lem:flownetwork-properties}.

If, for every~$t\in\NN$, the corresponding flow network~$\text{FN}(V^t,A^t,v^{t+1},\pi)$ admits a flow of value~$1$, the recursively defined partial network flow method $(y^t)_{t\in[T]}$ extends to a \emph{network flow method} $(y^t)_{t\in\NN}$.
By \Cref{lem:flownetwork-properties}, such a method satisfies global quota and ex-ante proportionality.
The next lemma guarantees that, for any number of parties~$n\leq 3$ and partial network flow method~$(y^t)_{t\in[T]}$, the corresponding flow network~$\text{FN}(V^T,A^T,v^{T+1},\pi)$ always admits a feasible flow of value~$1$.
Thus, for any~$T\in\NN$, the partial network flow method~$(y^t)_{t\in[T]}$ can be extended to a partial network flow method~$(y^t)_{t\in[T+1]}$.
Consequently, a partial network flow method yields a valid network flow method.
\begin{restatable}{lemma}{lemflowpolytope}\label{lem:flowpolytope}
    Let~$n\leq 3$ and let~$(y^k)_{k\in[t]}$ be a partial network flow method with corresponding flow network~$\normalfont\text{FN}(V^t,A^t,v^{t+1},\pi)$ for step~$t+1$.
    Then this flow network admits a feasible~$(o,d)$-flow of value~$1$, i.e., the partial network flow method can be extended to~$(y^k)_{k\in[t+1]}$.
\end{restatable}

The proof of \Cref{lem:flowpolytope} is deferred to \Cref{app:lem:flowpolytope}.
It proceeds by a case distinction over the possible configurations of cumulative allocation vectors and vote increments for the three parties.
Feasibility is shown either by explicitly constructing a valid flow or by reducing the network to a transshipment problem on a bipartite graph.
This ensures that in every case the flow network admits a feasible flow of value~$1$, completing the recursive construction.
\Cref{thm:exante} now follows by combining \Cref{lem:flownetwork-properties,lem:flowpolytope}.

Having shown that network flow methods satisfy all desired properties and are well-defined for up to three parties, we now establish that this family of methods in fact contains all online apportionment methods satisfying all three properties.

\begin{restatable}{proposition}{lemmethodsonetoone}\label{prop:methods-one-to-one}
    Any online apportionment method satisfying global quota and ex-ante proportionality is a network flow method.
\end{restatable}

In the proof, we construct a network flow for each step that replicates the decisions of a given online apportionment method satisfying global quota and ex-ante proportionality.
We verify that all capacity constraints are met, flow conservation is attained, and that the flow has a value of~$1$, which establishes that any such method can be obtained as a network flow method.
The full proof is deferred to \Cref{app:prop:methods-one-to-one}.

With the preceding proposition, we conclude that every online apportionment method satisfying global quota and ex-ante proportionality is a network flow method and can therefore be obtained via the construction described above.
Indeed, the proof of \Cref{lem:flowpolytope} provides an explicit recursive procedure for building the method as it defines the flow network at each step~$t$ based on the current allocation and voting vectors.

For the case of $n=3$ parties, it is worth mentioning that in all but two cases, the resulting flow of value~$1$ is uniquely determined.
Consequently, the seat assignment may only differ when either
\begin{enumerate}[label=(\alph*)]
    \item one seat is to be allocated and each cumulative allocation~$u\in U^t$ assigns the upper share of seats to a single party, while in the current step, no party receives enough votes to increase their upper share, or
    \item two seats are to be allocated, and each cumulative allocation~$u\in U^t$ assigns the upper share of seats to two parties, while in the current step, all parties receive enough votes to increase their upper share.
\end{enumerate}
In both cases, the flow network admits multiple feasible flows of value~$1$, leading to different valid seat allocations.

We conclude this section by discussing barriers to further positive results in the realm of randomized online apportionment methods.

\paragraph{Network construction and impossibility for $n\ge 4$.} 
The following example illustrates why the flow construction fails when $n\ge 4$, as the corresponding network no longer admits a feasible flow of value~$1$.
Consider an instance with four parties and voting vector~$v^1=(1/2,1/2,1/2,1/2)$.
A partial method assigns~$H^1=2$ seats in step~$1$.
Without any loss, assume that~$y^1(V^0,A^0,v^1,\{1,2\})=p>0$, i.e., parties~$1$ and~$2$ obtain a seat in the same allocation with strictly positive probability.

Now consider the voting vector~$v^2=(1/2,1/2,0,0)$.
In this case, for all~$i\in[n]$, we have~$\lceil V^1\rceil=\lceil V^1 +v^2\rceil$, so global quota implies that no party may receive an additional seat in step~$2$ if parties~$1$ and~$2$ received a seat each in step~$1$.
The flow network captures this by setting the upper capacity of all edges~$(\{1,2\},i)$ to zero.
Consequently, the edge~$(o,\{1,2\})$ must be assigned a flow of value~$0$ even though it has a strictly positive capacity of~$p$.
Since the total capacity of the outgoing edges of~$o$ equals~$1$, no flow of value~$1$ can exist.
This demonstrates that the construction of a partial network flow method fails at step~$t=2$.

\paragraph{Negative correlation.}
\citet{naor2025online} showed that, when $n=1$, a notion of negative correlation over time can be added on top of proportionality notions.
In our notation, negative correlation requires that for every party $i\in [n]$ and every pair of time steps $t,t'\in \NN$, \( \PP[a^t_i=1 \mid a^{t'}_i=1] \leq \PP[a^t_i=1]\).
The result of \citeauthor{naor2025online}\ implies that, when $n=1$, there exists an online apportionment method satisfying global quota, ex-ante proportionality, and negative correlation.\footnote{In fact, \citeauthor{naor2025online}\ used stronger forms of negative correlation, but we here explain why even this weaker notion is not possible with three parties.}
It is easy to see that this can be immediately extended to $n=2$ parties, since one party receives a seat if and only if the other does not.
However, it does not extend to three parties.
To see this, consider the following instance with vote vectors~$v^1=(0.5,0.5,0), v^2=(0,0.6,0.4),v^3=(0,0.2,0.8)$.
The outcome of the unique network flow method yields the following distribution over assignments:
\begin{align*}
    &\mathbb{P}[a_1^1=1,a_2^2=1,a_3^3=1]=0.5, \quad
    &\mathbb{P}[a_2^1=1,a_2^2=1,a_3^3=1]=0.1,\\
    &\mathbb{P}[a_2^1=1,a_3^2=1,a_2^3=1]=0.2, \quad
    &\mathbb{P}[a_2^1=1,a_3^2=1,a_3^3=1]=0.2.
\end{align*}
In particular, we observe that~$\mathbb{P}[a_2^3=1\vert a_2^1=0]=0<0.2=\mathbb{P}[a_2^3=1]$, which shows that negative correlation is violated.

\paragraph{Oblivious adversary.}
The impossibility result for instances with more than three parties relies on an adaptive adversary.
A natural question is how likely a method is to violate the quota axiom when facing an oblivious adversary.
Using the same four-party instance mentioned in the introduction, 
\citet{golz2025apportionment}
showed that any method violates quota with probability at least~$1/12$ since the undesirable allocation from the first two steps arises with this probability.
More generally, repeating this four-party construction~$\lfloor n/4 \rfloor$ times yields a violation probability of at least~$1-\big( \frac{11}{12}\big)^{\lfloor n/4\rfloor }$, which converges to~$1$ as~$n$ increases.

\section{Consequences in Online Dependent Rounding}\label{sec:connections}

Our results in Sections \ref{sec:deterministic} and \ref{sec:online} translate directly to the setting of online dependent rounding for multidimensional instances, i.e., we receive an $n$-dimensional marginal vector at each step instead of a single non-negative real value. 
In particular, the case of $n=1$ party captures the online level-set problem studied very recently by \cite{naor2025online}. 
Our Theorem \ref{thm:exante} shows the existence of an online dependent rounding algorithm with no deviation from the rounded cumulatives (global quota), and meeting the marginals for 3-dimensional instances (ex-ante proportionality).
Furthermore, in \Cref{prop:methods-one-to-one} we show that all such algorithms can be recovered via our network flow construction.
On the other hand, when we relax the condition of meeting the marginals on expectation, our greedy apportionment algorithm from Theorem \ref{thm:app-quota} is an online rounding algorithm for general $n$-dimensional instances that achieves optimal deviations from the rounded cumulatives.

In the rest of this section, and to showcase the applicability of our results in the context of online optimization, we introduce a multidimensional generalization of the multi-stage stochastic covering studied by \cite{naor2025online}, and we describe how to achieve near-feasibility and approximation guarantees as a simple consequence of our rounding results. 

\subsection{Multidimensional Multi-stage Hypergraph Stochastic Covering}
In the multidimensional multi-stage hypergraph stochastic covering problem (MMHSC), we are given a hypergraph $(N,A)$ with vertices $N$ and hyperedges $E$, and the hypergraph is $d$-uniform with $d\ge 2$, i.e., the size of each hyperedge is equal to $d$. 
We have $n$ types of resources, denoted by $[n]$, which are to be allocated on each step (stage) $t\in [T]$; the capacity of a node $u$ during stage $t$ is a non-negative integer number denoted by $C(u,t)$.
On the other hand, each hyperedge $e\in E$ has a demand of $D(i,t)$ for each resource $i$ and each step $t$, and $D(i,\cdot)$ is monotone non-decreasing.
A fractional solution for this problem corresponds to an allocation $y\colon N\times [n] \times [T]\to \RR_+$ satisfying the following:
\begin{align}
\sum_{u\in e}\sum_{\ell=1}^ty(u,i,\ell)&\ge D(i,t) \quad \text{ for each }i\in [n], t\in [T],\text{ and }e\in E,\label{eq:MHSC1}\\
\sum_{i=1}^ny(u,i,t)&\le C(u,t)\quad \text{ for each }u\in N\text{ and }t\in [T].\label{eq:MHSC2}
\end{align}
Constraints \eqref{eq:MHSC1} ensure that the total cumulative allocation for a hyperedge $e$ is at least the demand $D(i,t)$ for each resource $i$ and step $t$, and \eqref{eq:MHSC2} guarantee that on each step the total number of goods allocated to a vertex $u$ does not exceed the capacity $C(u,t)$. 
Then, given a fractional solution $y^{\star}$, at each step $t\in [T]$ the decision-maker receives the values $y^{\star}(u,i,t)$, i.e., the solution restricted to $t$.
The goal of the decision-maker is to round this solution online, i.e., at each step $t$, each value of the solution restricted to $t$ must be rounded to the floor or ceiling, and at the same time provide guarantees on the solution quality, i.e., objective value or feasibility.
In particular, our problem captures the multi-stage
stochastic hypergraph cover problem by \cite{naor2025online} when $n=1$, $D(1,t)=0$ for each $t<T$, $D(1,T)>0$, and $C(u,t)=\infty$ for each $u$ and $t$.
We remark that \cite{byrka2018approximation} study the graph case (i.e., $d=2$) and developed a 2-approximation algorithm, improving on the $2T$-approximation by \cite{swamy2012sampling}. 

In what follows, we show how to use Theorem \ref{thm:app-quota} and Theorem \ref{thm:exante} to construct integer solutions for MMHSC, online, with guarantees in terms of near-feasibility or objective approximation factor.
For the sake of exposition simplicity, suppose that the fractional solution $y^{\star}$ is binding on every constraint \eqref{eq:MHSC2}, i.e., there is no spare capacity over the vertices at every step.

\paragraph{Online near-feasible solutions.} 
We use Theorem \ref{thm:app-quota} to construct an online near-feasible solution for MMHSC.
For each vertex $u$ and every step $t$, consider $v_{i}^t(u)=y^{\star}(u,i,t)-\lfloor y^{\star}(u,i,t)\rfloor$ for each resource $i$.
At every $t$, for each vertex $u$ we use Theorem \ref{thm:app-quota} to get binary values $a^t_1(u),\ldots,a_n^t(u)$ satisfying the following properties: 
\begin{enumerate}[label=(\roman*)]
    \item $\sum_{i=1}^na_i^t(u)=C(u,t)-\sum_{i=1}^n\lfloor y^{\star}(u,i,t)\rfloor$,\label{det:ii}
    \item $|\sum_{\ell=1}^ta_i^{\ell}(u)-\sum_{\ell=1}^tv_i^{\ell}(u)|\le (n-1)/2$.\label{det:iii}
\end{enumerate}
At every $t$, the solution implemented by the decision-maker is equal to $Y(u,i,t)=\lfloor y^{\star}(u,i,t)\rfloor + a_i^t(u)$ for each $u$ and each $i$.
Observe that \ref{det:ii} directly implies that 
$\sum_{i=1}^nY(u,i,t)= C(u,t)$ for each $u\in N,$
i.e., \eqref{eq:MHSC2} is satisfied. 
On the other hand, for each $i\in [n]$ and $e\in E$, 
\begin{align*}
\sum_{u\in e}\sum_{\ell=1}^tY(u,i,t)
&=\sum_{u\in e}\left(\sum_{\ell=1}^t\lfloor y^{\star}(u,i,\ell)\rfloor+\sum_{\ell=1}^ta_i^{\ell}(u)\right)\\
&\ge \sum_{u\in e}\left(\sum_{\ell=1}^t\lfloor y^{\star}(u,i,\ell)\rfloor+\sum_{\ell=1}^tv_i^{\ell}(u)-\frac{n-1}{2}\right)\\
&= \sum_{u\in e}\left(\sum_{\ell=1}^ty^{\star}(u,i,\ell)-\frac{n-1}{2}\right)\ge D(i,t)-\frac{d(n-1)}{2},
\end{align*}
where the first inequality holds from \ref{det:iii} and the last inequality is a consequence of $y^{\star}$ satisfying \eqref{eq:MHSC1} and the hypergraph being $d$-uniform.
We conclude that $Y$ satisfies \eqref{eq:MHSC1} within an additive violation of $d(n-1)/2$.
In particular, when $n=3$, \Cref{cor:global-quota} implies that the violation is at most $d-1$.

\paragraph{Online approximate min-cost solutions.} Suppose now that the solution $y^{\star}$ received by the decision-maker implies a cost $\text{Cost}(y^{\star})$, where Cost is a linear function, and we have $n=3$ resources.
We use Theorem \ref{thm:exante} to construct an online $\alpha$-approximate solution for MHSC under resource augmentation, where $\alpha\coloneqq \max_{j\in \{1,2,3\},\ell\in [T]}(d+D(j,\ell)-1)/D(j,\ell)$.
For each vertex $u$ and every step $t$, consider $v_{i}^t(u)=\alpha y^{\star}(u,i,t)-\lfloor \alpha y^{\star}(u,i,t)\rfloor$ for each resource $i$.
At every $t$ and for each vertex $u$, we use Theorem \ref{thm:app-quota} to get binary values $a^t_1(u),\ldots,a_n^t(u)$ satisfying the following properties: 
\begin{enumerate}[label=(\roman*)]
    \item $\lfloor \alpha C(u,t)\rfloor-\sum_{i=1}^n\lfloor \alpha y^{\star}(u,i,t)\rfloor\le \sum_{i=1}^na_i^t(u)\le \lceil \alpha C(u,t)\rceil-\sum_{i=1}^n\lfloor \alpha y^{\star}(u,i,t)\rfloor$,\label{rand:ii}
    \item $\sum_{\ell=1}^ta_i^{\ell}(u)\in \Big\{\Big\lfloor\sum_{\ell=1}^tv_i^{\ell}(u)\Big\rfloor,\Big\lceil\sum_{\ell=1}^tv_i^{\ell}(u)\Big\rceil\Big\}$,\label{rand:iii}
    \item $\EE[a_i^t(u)]=v_i^t(u)$.\label{rand:ex}
\end{enumerate}
At every step $t$, the solution implemented by the decison-maker is equal to $Y(u,i,t)=\lfloor \alpha y^{\star}(u,i,t)\rfloor + a_i^t(u)$ for each $u$ and each $i$.
Observe that \ref{det:ii} directly implies that 
$\sum_{i=1}^nY(u,i,t)\le \lceil\alpha C(u,t)\rceil$ for each $u\in N,$
i.e., \eqref{eq:MHSC2} is satisfied under a resource augmentation of factor $\alpha$. 
On the other hand, for each resource $i\in [3]$, following the analysis of \cite{naor2025online} we can easily show that 
$\sum_{u\in e}\sum_{\ell=1}^tY(u,i,t)\ge D(i,t)$,
that is, $Y$ satisfies \eqref{eq:MHSC2}.
Finally, thanks to the ex-ante proportionality property \ref{rand:ex}, we conclude that 
$\EE[\text{Cost}(Y)]\le \alpha\cdot \text{Cost}(y^{\star})$.
\section{Concluding Remarks}

In this work, we have introduced an online version of the apportionment problem, where indivisible seats are to be allocated immediately and irrevocably to parties in every time step in proportion to the parties' votes, which are revealed online and can be set adaptively.
We have imposed a notion of local quota compliance upfront, requiring that each party receives, in each step, a number of seats equal to its fraction of votes, rounded up or down.
We have established that a greedy method provides the best-possible deviation, linear in $n$, between the cumulative votes and the cumulative seats that a party has received at any step.
This implies that a global quota axiom, requiring the same principle as local quota but for the cumulative values in each step, can be satisfied if and only if $n\leq 3$.
When this condition holds, we have shown that the proportionality guarantee given by global quota can be enhanced using randomization: For every instance, a lottery over methods satisfying global quota assigns to each party its proportional number of seats in expectation. 
Our proof is based on network flow techniques, providing a full characterization of such methods and allowing for efficient implementation.

Our work has natural connections and implications for related settings in online dependent rounding, as our results in Sections \ref{sec:deterministic} and \ref{sec:online} translate directly to online rounding guarantees for multidimensional instances where we receive an $n$-dimensional marginal vector at each step.
In particular, the case of $n=1$ party captures the online level-set problem studied very recently by \cite{naor2025online}. 

Several directions for future work arise.
Perhaps the more immediate one concerns the validity of \Cref{thm:exante} when global quota is replaced by approximate proportionality for $n\geq 4$. Specifically, one could aim to randomize over deterministic methods that attain the best possible proportionality guarantees from \Cref{thm:app-quota} while fulfilling ex-ante proportionality.

On the other hand, it is natural to ask whether lower expected deviations or stronger axiomatic properties can be achieved in an online model where the adversary is more restricted.
Two interesting possibilities include the study of randomized methods against an oblivious adversary, that cannot set the votes based on previous realizations of the method's randomness, and a fully stochastic model, where votes come from (potentially correlated) known distributions in each step.

\section*{Acknowledgements}

This work was partially supported by a Structural Democracy Fellowship through the Brooks School of Public Policy at Cornell University and by Centro de Modelamiento Matemático (CMM) BASAL fund FB210005 for center of excellence from ANID-Chile.

\newpage
\section*{Appendix}
\appendix

\section{Proofs Deferred from \Cref{sec:prelims}}

\subsection{Proof of \Cref{prop:offline}}\label{app:prop-offline}

\propOffline*

\begin{proof}
We fix $n\in \NN$ and an $n$-dimensional instance $(v^t)_{t\in [T]}$.
We consider a capacitated network~$G\coloneqq (P,E)$ given by
    \begin{enumerate}[label=\normalfont(\alph*)]
        \item vertex set~$P\coloneqq \{o,d\}\cup \{u^t\}_{t\in [T]} \cup \{w^t_i\}_{t\in [T],i\in [n]}$,
        \item edge set
        \begin{align*}
            E \coloneqq{} & \{(o,u^t):t\in [T]\}\cup \{(u^t,w^t_i):t\in [T],i\in [n]\}\\
            & \cup \{(w^t_i,w^{t+1}_i): t\in [T-1],i\in [n]\}\cup \{(w^T_i,d):i\in [n]\},
        \end{align*}
        and
        \item upper and lower capacities $c,\ell\colon E\to \RR_+$ given by
        \begin{align*}
            c((o,u^t)) =H^t,\quad
            & c((u^t,w^t_i)) = 1,\quad
            c((w^t_i,w^{t+1}_i)) = \lceil V^t_i\rceil,\quad
            \ell((w^t_i,w^{t+1}_i)) = \lfloor V^t_i \rfloor,\\
            & c((w^T_i,z)) = \lceil V^T_i\rceil,\quad
            \ell((w^T_i,d)) = \lfloor V^T_i \rfloor.
        \end{align*}
    \end{enumerate}

    For any feasible integral $(o,d)$-flow $f\colon E\to \NN_0$ on this network, we define the apportionment method $(X^t(f))_{t\in \NN}$ by setting $X^t(f)=\{i\in \NN: f((u^t,w^t_i))=1\}$ for every $t\in \NN$.
    \begin{claim}\label{claim:offline-global}
        For any feasible integral $(o,d)$-flow $f\colon E\to \NN_0$ on $G$, the offline apportionment method $(X^t(f))_{t\in [T]}$ satisfies global quota.        
    \end{claim}
    \begin{proof}
    To prove this claim, we simply observe that, for any feasible integral $(o,d)$-flow $f\colon E\to \NN_0$ and every $i\in [n]$, it holds
    \[
        |\{k\in [t]: i\in X^k(f)\} = \sum_{k=1}^t f((u^t,w^t_i)) = f((w^t_i,w^{t+1}_i)) \in \{\lfloor V^t_i\rfloor, \lceil V^t_i\rceil\}
    \]
    for every $t\in [T-1]$, and 
    \[
        |\{k\in [T]: i\in X^k(f)\} = \sum_{k=1}^T f((u^t,w^t_i)) = f((w^T_i,d)) \in \{\lfloor V^T_i\rfloor, \lceil V^T_i\rceil\}.
    \]
    Indeed, in both cases the first equality follows from the definition of $(X^t(f))_{t\in [T]}$, the second one from the flow conservation constraint at vertex $w^t_i$, and the inclusion from the capacity constraint on the edge $(w^t_i,w^{t+1}_i)$ or $(w^T_i,d)$.
    \end{proof}
    
    We now consider the fractional $(o,d)$-flow $f^*\colon E\to \RR_+$ defined as follows:
    \begin{align*}
        f^*((o,u^t))&=H^t \quad \text{for every }t\in [T],\\
        f^*((u^t,w^t_i)) & = v^t_i \quad \text{for every } t\in [T],\ i\in [n],\\
        f^*((w^t_i,w^{t+1}_i)) & = V^t_i \quad \text{for every } t\in [T-1],\ i\in [n],\\
        f^*((w^T_i,d)) & = V^T_i \quad \text{for every } i\in [n].
    \end{align*}
    The flow $f^*$ is easily seen to be a feasible $(o,d)$-flow on $G$.
    Therefore, because of the integrality of the network flow polytope~\citep{ahuja1988network}, it admits a convex decomposition into integer $(o,d)$-flows.
    That is, there exist an integer $q\in \NN$, integer $(o,d)$-flows $\{f_j\}_{j\in [q]}$, where $f_j\colon E\to \NN_0$ for every $j\in [q]$, and coefficients $\{\lambda_j\}_{j\in [q]}$, where $\lambda_j\in (0,1)$ for every $j\in [q]$ and $\sum_{j=1}^q \lambda_j=1$, such that
    \begin{equation}
        \sum_{j= 1}^q\lambda_jf_j(e) = f^*(e) \quad \text{for every } e \in E.\label{eq:convex-combination}
    \end{equation}
    We now consider the randomized offline apportionment method that runs each deterministic method $(X^t(f_j))_{t\in \NN}$ with probability $\lambda_j$, for $j\in [q]$.
    From \Cref{claim:offline-global}, we conclude that this method satisfies global quota.
    From Equation~\eqref{eq:convex-combination}, we conclude that it satisfies ex-ante proportionality.
\end{proof}

\section{Proofs Deferred from \Cref{sec:deterministic}}

\subsection{Proof of \Cref{lem:greedy}}\label{app:lem-greedy}

\lemGreedy*

\begin{proof}
	We begin with the case of general~$n$.
    We let, for the sake of contradiction, $(v^k)_{k\in \NN}$ be an $n$-dimensional instance and $t\in \NN$ be a time step such that $|s^t_i|> (n-1)/2$ for some $i\in [n]$.
    Since the parties are ordered from larger to smaller surplus, this implies either $s^t_1> (n-1)/2$ or $s^t_n < -(n-1)/2$; we will reach a contradiction in either case.

    We first consider the case where $s^t_1> (n-1)/2$, and introduce some additional notation.
    For each $k\in [t]_0$ and $\ell\in [n]$, we let
    $\sigma^k_\ell \coloneqq \frac{1}{\ell}\sum_{i\in [\ell]}s^k_i$ denote the average surplus of the $\ell$ parties with the largest surplus and we define $\gamma^k_\ell \coloneqq \sigma^k_\ell + \ell/2$.
    Observe that, for every $k\in [t]_0$ and $i\in [n-1]$
    \begin{equation}\label{eq:surplus-eqv}
    \begin{aligned}
            \gamma^k_i > \gamma^k_{i+1} & \Longleftrightarrow \frac{1}{i}\sum_{j\in [i]}s^k_j + \frac{i}{2} > \frac{1}{i+1}\sum_{j\in [i+1]}s^k_j + \frac{i+1}{2} \\
            & \Longleftrightarrow \left( 1+\frac{1}{i}\right) \sum_{j\in [i]}s^k_j -\frac{i+1}{2} >  \sum_{j\in [i+1]}s^k_j
            \Longleftrightarrow  s^k_{i+1} < \sigma^k_i -\frac{i+1}{2}.
    \end{aligned}
    \end{equation}
    Note that $s^t_1> (n-1)/2$ implies $\gamma^t_1 > n/2=\gamma^t_n$.
    Conversely, we have $1/2=\gamma^0_1<\gamma^0_2<\cdots <\gamma^0_n = n/2$, because $s^0_i=0$ for every $i\in [n]$. 
    Since $\gamma_n^k=n/2$ for all $k\in[t]_0$, there exists $k\in [t]$ such that $\gamma^k_i > n/2$ and $\gamma^k_i> \gamma^k_{i+1}$ for some $i\in [n-1]$; let $k^*$ be the minimum such $k$.
    Once $k^*$ is fixed, we fix $i^*$ to the minimum $i\in [n-1]$ such that $\gamma^{k^*}_{i^*}> n/2$ and $\gamma^{k^*}_{i^*}> \gamma^{k^*}_{i^*+1}$.
    The following claim will allow us to reach a contradiction.
    \begin{claim}\label{claim:surplus-diff}
        It holds $s^{k^*}_{i^*}\leq s^{k^*}_{i^*+1}+1$.
    \end{claim}
    \begin{proof}
        We begin by noting that $\gamma^{k^*}_{i^*}>\gamma^{k^*-1}_{i^*}$.
        If this was not true, we would have $\gamma^{k^*-1}_{i^*} \geq \gamma^{k^*}_{i^*} > n/2$, and thus $\gamma^{k^*-1}_{i}> \gamma^{k^*-1}_{i+1}$ for some $i\in [n-1]$, contradicting the choice of $k^*$.
        The inequality $\gamma^{k^*}_{i^*}>\gamma^{k^*-1}_{i^*}$ implies $\sigma^{k^*}_{i^*}>\sigma^{k^*-1}_{i^*}$, which in turn implies $\sum_{j=1}^{i^*} s^{k^*}_j > \sum_{j=1}^{i^*} s^{k^*-1}_j$.
        Since 
        \[
            \sum_{j=1}^{n} s^{k^*}_j = \sum_{j=1}^{n} s^{k^*-1}_j = 0,
        \]
        there must exist parties $i\in [i^*]$ and $j\in \{i^*+1,\ldots,n\}$ such that the surplus of the former increases from step~$k^*-1$ to step~$k^*$ and that of the latter decreases.
        Formally, if $\hat{s}_i$ is the surplus of party $i$ in step $k^*-1$  and $\hat{s}_j$ is the surplus of party $j$ in step $k^*-1$ (which are not necessarily equal to $s^{k^*-1}_i$ and $s^{k^*-1}_j$ due to reordering), it holds  $s^{k^*}_i>\hat{s}_i$ and $s^{k^*}_j<\hat{s}_j$.
        Since $\hat{s}_\ell=s^{k^*}_\ell+(a^{k^*}_\ell-v^{k^*}_\ell)$ for each $\ell\in [n]$, these inequalities yield  $a^{k^*}_i>v^{k^*}_i$ and $a^{k^*}_j<v^{k^*}_j$.
        Since $v^{k^*}_\ell\in [0,1)$ and $a^{k^*}_\ell\in \{0,1\}$ for each $\ell\in [n]$, we conclude that $a^{k^*}_i=1$ and $a^{k^*}_j=0<v^{k^*}_j$.
        
        Suppose now that the statement is not true, i.e., $s^{k^*}_{i^*}>s^{k^*}_{i^*+1}+1$.
        Since $i\in [i^*]$, $j \in \{i^*+1,\ldots,n\}$, and the parties are ordered from larger to smaller surplus, this implies 
        \[
        	s^{k^*}_{i}>s^{k^*}_{j}+1 \Longleftrightarrow \hat{s}_{i}+(a^{k^*}_i-v^{k^*}_i) >\hat{s}_{j}+(a^{k^*}_j-v^{k^*}_j)+1 \Longrightarrow \hat{s}_{i}-v^{k^*}_i >\hat{s}_{j}-v^{k^*}_j.
        \]
        But since the greedy method assigns seats, in step $k^*$, to the parties $\ell$ with the smallest value of $\hat{s}_\ell-v^{k^*}_\ell$ among those that receives positive votes, this contradicts the above statement that $a^{k^*}_i=1$ and $a^{k^*}_j=0<v^{k^*}_j$.
    \end{proof}
    We now reach an immediate contradiction when $i^*=1$, since in this case Property~\eqref{eq:surplus-eqv} implies $s^{k^*}_{2} < s^{k^*}_{1}-1$ and \Cref{claim:surplus-diff} states the opposite.
    When $i^*\geq 2$, we obtain
    \[
        s^{k^*}_{i^*}\leq s^{k^*}_{i^*+1}+1 < \sigma^{k^*}_{i^*}-\frac{i^*+1}{2}+1 = \frac{(i^*-1)\sigma^{k^*}_{i^*-1}+s^{k^*}_{i^*}}{i^*} - \frac{i^*-1}{2}.
    \]
    Indeed, the first inequality comes from \Cref{claim:surplus-diff} and the second one from Property~\eqref{eq:surplus-eqv}; the equality follows by splitting the average $\sigma^{k^*}_{i^*}$ of the first $i^*$ terms into the average $\sigma^{k^*}_{i^*-1}$ of the first $i^*-1$ terms and the $i^*$th term $s^{k^*}_{i^*}$.
    Rearranging terms, we obtain $s^{k^*}_{i^*}< \sigma^{k^*}_{i^*-1}-i^*/2$.
    But then, Property~\eqref{eq:surplus-eqv} implies $\gamma^{k^*}_{i^*-1}> \gamma^{k^*}_{i^*}>n/2$,
    contradicting the choice of $i^*$ as the minimum $i\in [n-1]$ such that $\gamma^{k^*}_{i}>n/2$ and $\gamma^{k^*}_{i}> \gamma^{k^*}_{i+1}$.

    We now consider the case where $s^t_n< -(n-1)/2$; the proof is mostly analogous to the previous case.
    We redefine the notation for this case.
    For each $k\in [t]_0$ and $\ell\in [n]$, we let
    $\sigma^k_\ell \coloneqq \frac{1}{\ell}\sum_{i=n+1-\ell}^n s^k_i$ denote the average surplus of the $\ell$ parties with the smallest surplus and we define $\gamma^k_\ell \coloneqq \sigma^k_\ell - \ell/2$.
    Observe that, for every $k\in [t]_0$ and $i\in [n-1]$
    \begin{equation}\label{eq:deficit-eqv}
    \begin{aligned}
            \gamma^k_i < \gamma^k_{i+1} & \Longleftrightarrow \frac{1}{i}\sum_{j=n+1-i}^n s^k_j - \frac{i}{2} < \frac{1}{i+1}\sum_{j=n-i}^n s^k_j - \frac{i+1}{2} \\
            & \Longleftrightarrow \left( 1+\frac{1}{i}\right) \sum_{j=n+1-i}^n s^k_j +\frac{i+1}{2} <  \sum_{j=n-i}^n s^k_j
            \Longleftrightarrow  s^k_{n-i} > \sigma^k_i +\frac{i+1}{2}.
    \end{aligned}
    \end{equation}
    Note that $s^t_n< -(n-1)/2$ implies $\gamma^t_1 < -n/2=\gamma^t_n$.
    Conversely, we have $-1/2=\gamma^0_1>\gamma^0_2>\cdots >\gamma^0_n = -n/2$, because $s^0_i=0$ for every $i\in [n]$.
    Since $\gamma^k_n=-n/2$ for every $k\in [t]_0$, there exists $k\in [t]$ such that $\gamma^k_i < -n/2$ and $\gamma^k_i< \gamma^k_{i+1}$ for some $i\in [n-1]$; let $k^*$ be the minimum such $k$.
    Once $k^*$ is fixed, we fix $i^*$ to the minimum $i\in [n-1]$ such that $\gamma^{k^*}_{i^*}< -n/2$ and $\gamma^{k^*}_{i^*}< \gamma^{k^*}_{i^*+1}$.
    The following claim will allow us to reach a contradiction.
    \begin{claim}\label{claim:deficit-diff}
        It holds $s^{k^*}_{n+1-i^*}\geq s^{k^*}_{n-i^*}-1$.
    \end{claim}
    \begin{proof}
        We begin by noting that $\gamma^{k^*}_{i^*}<\gamma^{k^*-1}_{i^*}$.
        If this was not true, we would have $\gamma^{k^*-1}_{i^*} \leq \gamma^{k^*}_{i^*} < -n/2$, and thus $\gamma^{k^*-1}_{i}< \gamma^{k^*-1}_{i+1}$ for some $i\in [n-1]$, contradicting the choice of $k^*$.
        The inequality $\gamma^{k^*}_{i^*}<\gamma^{k^*-1}_{i^*}$ implies $\sigma^{k^*}_{i^*}<\sigma^{k^*-1}_{i^*}$, which in turn implies $\sum_{j=n+1-i^*}^{n} s^{k^*}_j < \sum_{j=n+1-i^*}^{n} s^{k^*-1}_j$.
        Since 
        \[
            \sum_{j=1}^{n} s^{k^*}_j = \sum_{j=1}^{n} s^{k^*-1}_j = 0,
        \]
        there must exist parties $i\in \{n+1-i^*,\ldots,n\}$ and $j\in [n-i^*]$ such that the surplus of the former decreases from $k^*-1$ to $k$ and that of the latter increases.
        Formally, if $\hat{s}_i$ is the surplus of party $i$ in step $k^*-1$ and $\hat{s}_j$ is the surplus of party $j$ in step $k^*-1$ (which are not necessarily equal to $s^{k^*-1}_i$ and $s^{k^*-1}_j$ due to reordering), it holds $s^{k^*}_i<\hat{s}_i$ and $s^{k^*}_j>\hat{s}_j$.
        Since $\hat{s}_\ell=s^{k^*}_\ell+(a^{k^*}_\ell-v^{k^*}_\ell)$ for each $\ell\in [n]$, these inequalities yield  $a^{k^*}_i<v^{k^*}_i$ and $a^{k^*}_j>v^{k^*}_j$.
        Since $v^{k^*}_\ell\in [0,1)$ and $a^{k^*}_\ell\in \{0,1\}$ for each $\ell\in [n]$, we conclude that $a^{k^*}_i=0<v^{k^*}_i$ and $a^{k^*}_j=1$.
        
        Suppose now that the statement is not true, i.e., $s^{k^*}_{n+1-i^*}<s^{k^*}_{n-i^*}-1$.
        Since $i\in \{n+1-i^*,\ldots,n\}$, $j \in [n-i^*]$, and the parties are ordered from larger to smaller surplus, this implies 
        \[
        	s^{k^*}_{i}<s^{k^*}_{j}-1 \Longleftrightarrow \hat{s}_{i}+(a^{k*}_i-v^{k^*}_i)<\hat{s}_{j}+(a^{k^*}_j-v^{k^*}_j) - 1 \Longrightarrow \hat{s}_{i}-v^{k^*}_i<\hat{s}_{j}-v^{k^*}_j.
        \]
        But since the greedy method assigns seats, in step $k^*$, to the parties~$\ell$ with the smallest value of $\hat{s}_\ell-v^{k^*}_\ell$ among those that receives positive votes, this contradicts the above statement that $a^{k^*}_i=0<v^{k^*}_i$ and $a^{k^*}_j=1$.
    \end{proof}
    
    We now reach an immediate contradiction when $i^*=1$, since in this case Property~\eqref{eq:deficit-eqv} implies $s^{k^*}_{n-1} > s^{k^*}_{n}+1$ and \Cref{claim:deficit-diff} states the opposite.
    When $i^*\geq 2$, we obtain
    \[
        s^{k^*}_{n+1-i^*}\geq s^{k^*}_{n-i^*}-1 > \sigma^{k^*}_{i^*}+\frac{i^*+1}{2}-1 = \frac{(i^*-1)\sigma^{k^*}_{i^*-1}+s^{k^*}_{n+1-i^*}}{i^*} + \frac{i^*-1}{2}.
    \]
    Indeed, the first inequality comes from \Cref{claim:deficit-diff} and the second one from Property~\eqref{eq:deficit-eqv}; the equality follows by splitting the average $\sigma^{k^*}_{i^*}$ of the first $i^*$ terms into the average $\sigma^{k^*}_{i^*-1}$ of the first $i^*-1$ terms and the $i^*$th term $s^{k^*}_{n+1-i^*}$.
    Rearranging terms, we obtain $s^{k^*}_{n+1-i^*}> \sigma^{k^*}_{i^*-1}+i^*/2$.
    But then, Property~\eqref{eq:deficit-eqv} implies $\gamma^{k^*}_{i^*-1}< \gamma^{k^*}_{i^*}<-n/2$,
    contradicting the choice of $i^*$ as the minimum $i\in [n-1]$ such that $\gamma^{k^*}_{i}< -n/2$ and $\gamma^{k^*}_{i}< \gamma^{k^*}_{i+1}$.
    
    We now address strict $1$-proportionality for $n=3$.
    We note that, in general, if we aim to prove strict proportionality in the cases above, we can proceed analogously.
    That is, we suppose towards a contradiction the existence of an $n$-dimensional instance $(v^t)_{t\in \NN}$ and a step $t\in \NN$ such that either $s^t_1\geq (n-1)/2$ or $s^t_n\leq -(n-1)/2$ holds, and define the values $\sigma^k_\ell$ and $\gamma^k_\ell$ as before, and fix $k^*$ and $i^*$ analogously, but now for the corresponding weak inequalities ($\gamma^{k^*}_{i^*}\geq n/2$ and $\gamma^{k^*}_{i^*}\geq \gamma^{k^*}_{i^*+1}$ for the first case; $\gamma^{k^*}_{i^*}\leq -n/2$ and $\gamma^{k^*}_{i^*}\leq \gamma^{k^*}_{i^*+1}$ for the second case).
    Properties~\eqref{eq:surplus-eqv} and~\eqref{eq:deficit-eqv} admit analogous versions with weak inequalities, whereas \Cref{claim:surplus-diff,claim:deficit-diff} remain true and can be proven in a completely analogous way, as they rely on the respective observations $\gamma^{k^*}_{i^*}>\gamma^{k^*-1}_{i^*}$ and $\gamma^{k^*}_{i^*}<\gamma^{k^*-1}_{i^*}$, which remain true.
    Furthermore, the contradictions for the case $i^*\geq 2$ can be reached in the same way as before.
    Indeed, for the case $s^t_1\geq (n-1)/2$ we obtain
    \[
    s^{k^*}_{i^*}\leq s^{k^*}_{i^*+1}+1 \leq \sigma^{k^*}_{i^*}-\frac{i^*+1}{2}+1 = \frac{(i^*-1)\sigma^{k^*}_{i^*-1}+s^{k^*}_{i^*}}{i^*} - \frac{i^*-1}{2},
    \]
    which implies $s^{k^*}_{i^*}\leq \sigma^{k^*}_{i^*-1}-i^*/2$ and thus $\gamma^{k^*}_{i^*-1}\geq \gamma^{k^*}_{i^*}\geq n/2$ due to Property~\eqref{eq:surplus-eqv},
    contradicting the choice of $i^*$.
    Similarly, for the case $s^t_n\leq -(n-1)/2$ we obtain
    \[
    s^{k^*}_{n+1-i^*}\geq s^{k^*}_{n-i^*}-1 \geq \sigma^{k^*}_{i^*}+\frac{i^*+1}{2}-1 = \frac{(i^*-1)\sigma^{k^*}_{i^*-1}+s^{k^*}_{n+1-i^*}}{i^*} + \frac{i^*-1}{2},
    \]
    which implies $s^{k^*}_{n+1-i^*}\geq \sigma^{k^*}_{i^*-1}+i^*/2$ and thus $\gamma^{k^*}_{i^*-1}\leq \gamma^{k^*}_{i^*}\leq -n/2$ due to Property~\eqref{eq:deficit-eqv}, contradicting the choice of $i^*$.
    
    When $i^*=1$, \Cref{claim:surplus-diff,claim:deficit-diff} together with the weak versions of Properties~\eqref{eq:surplus-eqv} and~\eqref{eq:deficit-eqv} do not yield a contradiction anymore.
    However, they now imply that $s^{k^*}_2=s^{k^*}_1-1$ when we assume $s^t_1\geq (n-1)/2$ or that $s^{k^*}_{n-1}=s^{k^*}_n+1$ when we assume $s^t_n\leq -(n-1)/2$.
    We will see that these conditions still yield a contradiction when $n=3$, so we fix this value.
    Note that, when we start from the assumption that $s^t_1\geq 1$ for some $t\in \NN$, $\gamma^{k^*}_1\geq 3/2$ is equivalent $s^{k^*}_1\geq 1$.
    Since $\sum_{i\in [3]}s^{k^*}_i = 0$, $s^{k^*}_2=s^{k^*}_1-1$ and $s^{k^*}_1\geq 1$ together imply $s^{k^*}_3\leq -1$ in this case.
    Similarly, when we start from the assumption that $s^t_3\leq -1$ for some $t\in \NN$, $\gamma^{k^*}_3\leq -3/2$ is equivalent to $s^{k^*}_3\leq -1$.
    Since $\sum_{i\in [3]}s^{k^*}_i = 0$, $s^{k^*}_2=s^{k^*}_3+1$ and $s^{k^*}_3\leq -1$ together imply $s^{k^*}_1\geq 1$ in this case.
    Thus, we conclude in either case that $s^{k^*}_1\geq 1 \Leftrightarrow s^{k^*}_3\leq -1$.
    
    Letting $\hat{s}_1$ denote the surplus of party $1$ in step $k^*-1$ and $\hat{s}_3$ the surplus of party $3$ in step $k^*-1$ (which are not necessarily equal to $s^{k^*-1}_1$ and $s^{k^*-1}_3$ due to reordering), we also know from the definition of $k^*$ that $\hat{s}_1<1$ and $\hat{s}_3>-1$, which will yield a similar contradiction as in the proofs of \Cref{claim:surplus-diff,claim:deficit-diff}.
    Indeed, $s^{k^*}_1>\hat{s}_1$ implies $a^{k^*}_1=1$ and $s^{k^*}_3<\hat{s}_3$ implies $a^{k^*}_3=0<v^{k^*}_3$.
    On the other hand, $s^{k^*}_1\geq 1$ implies $\hat{s}_1-v^{k^*}_1\geq 0$ and $s^{k^*}_3\leq -1$ implies $\hat{s}_3-v^{k^*}_3\leq -1$, so that $\hat{s}_3-v^{k^*}_3<\hat{s}_1-v^{k^*}_1$.
    But since the greedy method assigns seats, in step $k^*$, to the parties $i$ with the smallest value of $\hat{s}_i-v^{k^*}_i$ among those that receive positive votes, this contradicts the above statement that $a^{k^*}_1=1$ and $a^{k^*}_3=0<v^{k^*}_3$.
\end{proof}

\subsection{Proof of \Cref{lem:splitters}}\label{app:lem-splitters}

\lemSplitters*

\begin{proof}
    We let $n,t\in \NN,V^t,(X^k)_{k\in \NN},A^t,i,j$ be as in the statement.
    We define $v^{t+1}\in [0,1)^n$ by
    \[
        v^{t+1}_k = \begin{cases}
            \frac{1+(s^{t}_i-s^{t}_j)}{2} & \text{if } k=i,\\
            \frac{1-(s^{t}_i-s^{t}_j)}{2} & \text{if } k=j,\\
            0 & \text{otherwise,}
        \end{cases} \qquad \text{for all } k\in [n],
    \]
    whose entries sum up to $1$ and is, therefore, a valid vote vector.
    Since $H^{t+1}=\sum_{k\in [n]}v^{t+1}_k=1$ and  $v^{t+1}_k=0$ for every $k\notin \{i,j\}$, we have either $X^{t+1}=\{i\}$ or $X^{t+1}=\{j\}$.
    Since
    \[
        s^t_i-v^{t+1}_i = \frac{s^{t}_i+s^{t}_j-1}{2} = s^t_j-v^{t+1}_j,
    \]
    we have that in either case, some $k\in \{i,j\}$ is such that $s^{t+1}_k=\frac{s^{t}_i+s^{t}_j-1}{2}+1 = \frac{s^{t}_i+s^{t}_j+1}{2}$ and the other party $\ell\in \{i,j\}\setminus \{k\}$ is such that $s^{t+1}_\ell = \frac{s^{t}_i+s^{t}_j-1}{2}$.
    The claim follows.
\end{proof}

\subsection{Proof of \Cref{lem:prop-lb}}\label{app:lem-prop-lb}

\lemPropLB*

\begin{proof}
    We fix $n\in \NN$ and prove the lemma by induction over $n'$, the size of the subset of parties we consider. 
    Throughout this proof, whenever we have an $n'$-dimensional vote vector with entries in $P\subseteq [n]$, namely $v'\in [0,1]^{P}$, we refer to the vector $v\in [0,1]^{n}$ defined as 
    \[
        v_i = \begin{cases}
            v'_i & \text{if } i\in P\\
            0 & \text{otherwise,}
        \end{cases}
    \]
    as the \textit{$n$-dimensional extension} of $v'$. 
    This definition naturally carries over for the allocation vector.
    Note that, from our definition of an online apportionment method, all entries of these vote and allocation vectors that do not correspond to parties in $P$ are always fixed to $0$.
    
    Since we use induction steps from $n'$ to $n'+2$, we establish the validity of the lemma for two base cases, namely $n'=1$ and $n'=2$.
    The case $n'=1$ is trivial, as the conclusion of the lemma only requires $|s^t_i|\geq -\varepsilon$, and we know that $|s^t_i|\geq 0>-\varepsilon$.
    For the case $n'=2$, we fix $r, \varepsilon$, $(v^k)_{k\in[r]}$, $(X^t)_{t\in \NN}$, and $A^r$ as in the statement, we denote the parties by $P=\{p_1,p_2\}$, and we assume without loss of generality that $s^r_{p_1}\geq s^r_{p_2}$.
    If $s^r_1\geq 1/2-\varepsilon$ or $s^r_n\leq -(1/2-\varepsilon)$, we are done.
    Otherwise, it holds that~$s^{r}_1-s^{r}_2<1$, so we can consider a single step given by the $n$-dimensional extension of the $(p_1,p_2)$-splitter $\spl(p_1,p_2)$.
    \Cref{lem:splitters} implies that $s^{r+1}_{p_1}-s^{r+1}_{p_2}=1$, and thus either $s^{r+1}_{p_1}\geq 1/2$ or $s^{r+1}_{p_2} \leq -1/2$ holds.

    If $n\leq 2$, the previous cases suffice to conclude, so we assume the opposite in what follows.
    We assume that the statement of the lemma holds for some fixed $n'\in [n-2]$.
    We will make use of the step whose existence is guaranteed by this inductive hypothesis repeatedly, so we introduce some slightly more general notation for it.
    We consider tuples $(P,\varepsilon,(v^k)_{k\in [r]},A^r)$ such that $P=\{p_1,\ldots,p_{n'}\}$ is a set of $n'$ parties, $\varepsilon>0$ is a positive number, $(v^k)_{k\in [r]}$ is a partial $n$-dimensional instance, and $A^r=\sum_{k\in [r]} a^k$ is the aggregated allocation up to step $r$.
    Without loss of generality, we assume, as usual, that $s^{r}_{p_1}\geq s^r_{p_2}\geq \cdots \geq s^r_{p_{n'}}$.
    Given such a tuple, we denote the $n$-dimensional extension of the vote sequence satisfying the properties of the lemma by $\boost(P,\varepsilon,(v^k)_{k\in [r]},A^r)$, and refer to it as a \textit{$(P,\varepsilon)$-booster} (because it \textit{boosts} the surplus or deficit of some party in $P$).
    That is, the $(P,\varepsilon)$-booster $\boost(P,\varepsilon,(v^k)_{k\in [r]},A^r)$ is a partial $n$-dimensional instance $(v^k)_{k\in \{r+1,\ldots,t\}}$, with $t\geq r$ finite and $v^k\in [0,1]^n$ for $k\in \{r+1,\ldots,t\}$, such that for any allocation $(a^k)_{k\in \{r+1,\ldots,t\}}$ produced by an online apportionment method, with $a^k\in \{0,1\}^n$ for $k\in \{r+1,\ldots,t\}$, it holds that $s^{t}_{p_1}\geq (n'-1)/2-\varepsilon$ or $s^{t}_{p_{n'}} \leq -((n'-1)/2-\varepsilon)$ (potentially after reordering $p_1,\ldots,p_{n'}$).
        
    To prove the claim for $n'+2$, we fix $r, \varepsilon$, $V^r$, and $A^r$ as in the statement.
    We assume, as usual, that $s^r_{p_1}\geq s^r_{p_2}\geq \cdots \geq s^r_{p_{n'+2}}$.
    We apply the inductive hypothesis on the $n'$-dimensional instance that results from restricting to parties $p_2,\ldots,p_{n'+1}$, but with a smaller deviation of $\varepsilon/2$.
    That is, we consider the $n$-dimensional vote sequence 
    \[
        (v^k)_{k\in \{r+1,\ldots,t\}} = \boost\left(\{p_2,\ldots,p_{n'+1}\},\varepsilon/2,(v^k)_{k\in [r]},A^r\right)
    \]
    for some finite $t$.
    After reordering the parties in $P$ so that $s^t_{p_1}>s^t_{p_2}>\cdots >s^t_{p_{n'+2}}$, we know by the induction hypothesis that we have reached an instance with $s^{t}_{p_1} \geq (n'-1)/2-\varepsilon/2$ or $s^{t}_{p_{n'+2}} \leq -((n'-1)/2-\varepsilon/2)$.
    We again apply a $(\{p_2,\ldots,p_{n'+1}\},\varepsilon/2)$-booster, and then a third time if necessary.
    After applying this procedure at most three times, we now have obtained an instance, say at step $t'$, such that $s^{t'}_{p_2} \geq (n'-1)/2-\varepsilon/2$ or $s^{t'}_{p_{n'+1}} \leq -((n'-1)/2-\varepsilon/2)$, since each booster fixes a new party satisfying one of these two conditions.
        
    In the remainder of the proof, we repeatedly apply either a $(p_1,p_2)$-splitter or a $(p_{n'+1},p_{n'+2})$-splitter, and then a $\{p_2,\ldots,p_{n+1}\}$-booster.
    The splitters will increase the surplus of party $p_1$ or the deficit of party $p_{n'+2}$, at the cost of decreasing the surplus of $p_2$ or the deficit of $p_{n'+1}$, respectively.
    However, the inductive hypothesis allows us to boost either this surplus or deficit so that it surpasses $(n'-1)/2-\varepsilon/2$ and can be used for a splitter again.
    By repeating this procedure, we can make the surplus of party $p_1$ or the deficit of party $p_{n+2}$ arbitrarily close to $(n'-1)/2+1-\varepsilon/2$, and thus surpass $(n'+1)/2-\varepsilon$ after a finite number of steps.
    
        The argument can be formalized through the following algorithm.
        We initialize $\ell=0$, $t_\ell=t'$ and repeat the following procedure until $s^{t_\ell}_1 \geq (n'+1)/2-\varepsilon_1$ or $s^{t_\ell}_{n'+2} \leq (n'+1)/2-\varepsilon_1$:
        \begin{enumerate}[label=(\roman*)]
            \item If $s^{t_\ell}_{p_2} \geq (n'-1)/2-\varepsilon/2$, apply a $(p_1,p_2)$-splitter; i.e., take $v^{t_\ell+1}=\spl(p_1,p_2)$;
            \item else, if $s^{t_\ell}_{p_{n'+1}} \leq -((n'-1)/2-\varepsilon/2)$, apply a $(p_{n'+1},p_{n'+2})$-splitter; i.e.,  $v^{t_\ell+1}=\spl(p_{n'+1},p_{n'+2})$;
            \item apply a $\{p_2,\ldots,p_{n'+1}\}$-booster; i.e., take
            \[
                (v^k)_{k\in \{t_\ell+2,\ldots,t''\}} = \boost\left(\{p_2,\ldots,p_{n'+1}\},\varepsilon/2,(v^k)_{k\in [t_\ell+1]},A^{t_\ell+1}\right)
            \]
            for some finite $t''$;
            \item update $\ell\leftarrow \ell+1$ and $t_\ell\leftarrow t''$.
        \end{enumerate}

        Note that, after each iteration $\ell\geq 1$, the following invariant (which we showed for $t_0$) is preserved by the induction hypothesis: 
        $s^{t_\ell}_{p_2} \geq (n'-1)/2-\varepsilon/2$ or $s^{t_\ell}_{p_{n'+1}} \leq -((n'-1)/2-\varepsilon/2)$.
        Observe that $s^{t_\ell}_{p_1} < (n'+1)/2-\varepsilon$ and $s^{t_\ell}_{p_{n'+2}} > -((n'+1)/2-\varepsilon\big)$ as long as the algorithm does not terminate.
        Thus, whenever $s^{t_\ell}_{p_2} \geq (n'-1)/2-\varepsilon/2$ holds, we have $s^{t_\ell}_{p_1}-s^{t_\ell}_{p_2}<1$ and thus the $(p_1,p_2)$-splitter is well defined.
        Similarly,  whenever $s^{t_\ell}_{p_{n'+1}} \leq -((n'-1)/2-\varepsilon/2)$ holds, we have $s^{t_\ell}_{p_{n'+1}}-s^{t_\ell}_{p_{n'+2}}<1$ and thus the $(p_{n'+1},p_{n'+2})$-splitter is well defined.
        We distinguish between two types of iterations:
        \textit{type~I} are the iterations $\ell$ such that $s^{t_\ell}_{p_2} \geq (n'-1)/2-\varepsilon/2$ and the algorithm applies a $(p_1,p_2)$-splitter; analogously, \textit{type~II} are the iterations $\ell$ such that $s^{t_\ell}_{p_{n'+1}} \leq -((n'-1)/2-\varepsilon/2)$ and the algorithm applies a $(p_{n'+1},p_{n'+2})$-splitter.
        The following claim will allow us to conclude.
        \begin{claim}\label{claim:geom-bound}
        	Let $i(0),i(1),\ldots,i(L_1)\in \NN_0$ be the type I iterations and $j(0),j(1),\ldots,j(L_2)\in \NN_0$ be the type II iterations after $L_1+L_2+2$ iterations of the algorithm.
        	For every $\ell\in [L_1]_0$, it holds that $s^{t_{i(\ell)}}_{p_1}\geq (n'+1)/2-\varepsilon/2-1/2^\ell$.
        	For every $\ell\in [L_2]_0$, it holds that $s^{t_{j(\ell)}}_{p_{n'+2}}\leq -((n'+1)/2-\varepsilon/2-1/2^\ell)$.
        \end{claim}
        \begin{proof}
        Let $i(1),i(2),\ldots,i(L_1)$ and $j(1),j(2),\ldots,j(L_2)$ be as in the statement.
        The proofs for type I and type II iterations are fully symmetric, so we only include the former.	
        
        We proceed by induction.
        The base case $\ell=0$ is trivial, since at the beginning of the algorithm we have $s^{t_{0}}_{p_1} \geq (n'-1)/2-\varepsilon/2$, and $s^{t_{i(0)}}_{p_1}\geq s^{t_{0}}_{p_1}$.
        
        Assume the claim holds for some $\ell \in [L_1-1]_0$; i.e., $s^{t_{i(\ell)}}_{p_1}\geq (n'+1)/2-\varepsilon/2-1/2^\ell$.
        We will prove that 
        \begin{equation}
        	s^{t_{i(\ell)+1}}_{p_1}\geq \frac{n'+1}{2}-\frac{\varepsilon}{2}-\frac{1}{2^{\ell+1}}.\label{ineq:claim-iter-spl}
        \end{equation}
        If this is true, this suffices to conclude.
        Indeed, $s^{t_q}_{p_1}$ does not decrease in type II iterations $q$, since they only involve the parties $\{p_2,\ldots,p_{n'+2}\}$ in each iteration and the maximum surplus can only remain the same or increase.
        Thus, Inequality~\eqref{ineq:claim-iter-spl} implies
        \[
        s^{t_{i(\ell+1)}}_{p_1}\geq s^{t_{i(\ell)+1}}_{p_1}\geq \frac{n'+1}{2}-\frac{\varepsilon}{2}-\frac{1}{2^{\ell+1}}.
        \]
        
        We now prove Inequality~\eqref{ineq:claim-iter-spl} using the inductive hypothesis $s^{t_{i(\ell)}}_{p_1}\geq (n'+1)/2-\varepsilon/2-1/2^\ell$.
        If $s^{t_{i(\ell)}}_{p_1}\geq (n'+1)/2-\varepsilon/2-1/2^{\ell+1}$, then we are done, because the maximum surplus cannot decrease with a $(p_1,p_2)$-splitter.
        Otherwise, we let $\delta\coloneqq s^{t_{i(\ell)}}_{p_1}-((n'+1)/2-\varepsilon/2-1/2^\ell)$ and note that $0\leq \delta<1/2^{\ell+1}$.
        On the other hand, we know from the condition of a type I iteration that $s^{t_{i(\ell)}}_{p_2} \geq (n'-1)/2-\varepsilon/2$. 
        Combining these two facts, we obtain
        \[
            s^{t_{i(\ell)}}_{p_1}-s^{t_{i(\ell)}}_{p_2} \leq \frac{n'+1}{2}-\frac{\varepsilon}{2}-\frac{1}{2^\ell} +\delta - \left( \frac{n'-1}{2}-\frac{\varepsilon}{2}\right) = 1-\frac{1}{2^\ell} +\delta.
        \]
        Hence, the application of the $(p_1,p_2)$-splitter yields, due to \Cref{lem:splitters},
        \begin{align*}
            s^{t_{i(\ell)+1}}_{p_1} \geq s^{t_{i(\ell)}}_{p_1} + \frac{1-(s^{t_{i(\ell)}}_{p_1}-s^{t_{i(\ell)}}_{p_2})}{2} & \geq \frac{n'+1}{2}-\frac{\varepsilon}{2}-\frac{1}{2^\ell} +\delta + \frac{\frac{1}{2^\ell} -\delta}{2}\\
            & = \frac{n'+1}{2}-\frac{\varepsilon}{2}+\frac{\delta}{2} -\frac{1}{2^{\ell+1}} \geq  \frac{n'+1}{2}-\frac{\varepsilon}{2} -\frac{1}{2^{\ell+1}}.\qedhere
        \end{align*}
        \end{proof}
        
        After $L+2$ iterations of the algorithm, for even $L>0$, we know that we have either applied $L/2+1$ type I iterations or $L/2+1$ type II iterations.
        Thus, \Cref{claim:geom-bound} implies that either $s^{L}_{p_1} \geq (n'-1)/2-\varepsilon$ or $s^{L}_{p_{n'+2}} \leq -\big((n'-1)/2-\varepsilon\big)$ holds if
        \[
            \frac{n'+1}{2}-\frac{\varepsilon}{2} - \frac{1}{2^{L/2}} \geq \frac{n'+1}{2}-\varepsilon \Longleftrightarrow L\geq -2\log(\varepsilon/2).\qedhere
        \]
\end{proof}

\section{Proofs Deferred from \Cref{sec:online}}

\subsection{Proof of \Cref{lem:flownetwork-properties}\label{app:lem:flownetwork-properties}}

\lemflownetworkproperties*

\begin{proof}
    The proof follows by induction over~$t\in[T]$.
    First, consider step~$t=1$.
    By definition,~$y^1$ yields a valid probability distribution over subsets of size~$H^{1}$, ensuring that no party receives more than one seat in step~$1$.
    Thus, for proving global quota, it suffices to show that any party~$i$ with~$v_i^1=0$ is allocated a seat with probability zero.
    Let~$i$ be such that~$v_i^1=0$.
    Since~$0=v_i^1=\sum_{j=1}^m \lambda_j z_{j,i}$ with~$\lambda_j\in[0,1]$ for every~$j$, it follows that~$\lambda_j=0$ whenever~$z_{j,i}=1$.
    Hence,~$y^1(V^0,A^0,v^1,S)=0$ for all subsets~$S$ containing~$i$, as required.
    To prove that Equation~\eqref{eq:EAPdef} is also satisfied, we fix a party~$i\in [n]$ and observe that
    \[
    \sum_{S\subseteq [n]:i\in S}y^1(V^0,A^0,v^1,S)=\sum_{j=1}^m \lambda_j z_{j,i} =v_i^1.
    \]

    Next, consider some~$t\in [T-1]$ and assume that the partial method~$(y^t)_{t\in[T-1]}$ satisfies the required properties up to step~$t\leq T-1$.
    We show that they are maintained for step~$t+1$.
    By definition,~$y^{t+1}$ yields a valid probability distribution over subsets of size~$H^{t+1}$, ensuring that no party receives more than one seat in step~$t+1$.
    Consider a party~$i$ with~$v_i^{t+1}=0$.
    Since~$f$ is a feasible flow, we have~$f((u,i))=0$ for every~$u\in U^t$, which implies~$z_i(u)=0$.
    Consequently, each~$z_{j,i}=0$ and~$i\notin Z_j$ for all~$j$.
    Hence,~$y^{t+1}(V^t,A^t,v^{t+1},S)=0$ for all subsets~$S$ containing~$i$, as required.

    To verify global quota, fix a cumulative allocation~$A^t$ occurring with strictly positive probability and let~$u=u(A^t)\subset[n]$ be the corresponding set of parties at their upper share of seats.
    First, suppose that~$i\in u$ and that the upper share of~$i$ does not increase from step~$t$ to~$t+1$, i.e.,~$\lceil \sum_{k=1}^t v_i^k\rceil=\lceil \sum_{k=1}^{t+1} v_i^k\rceil$.
    Then party~$i$ cannot be allocated a seat in step~$t+1$ without violating global quota.
    Since~$f((u,i))=0$, we again have~$z_i(u)=0$, which implies that~$z_{j,i}=0$ and~$i\notin Z_j$ for all~$j$.
    Therefore,~$y^{t+1}(V^t,A^t,v^{t+1},S)=0$ for all subsets~$S$ containing~$i$, as required.

    On the other hand, assume~$i\notin u$ and the lower quota increases, i.e.,~$\lfloor \sum_{k=1}^t v_i^k\rfloor+1=\lfloor \sum_{k=1}^{t+1} v_i^k\rfloor$.
    Then~$i$ must be assigned a seat in step~$t+1$, as otherwise global quota would be violated.
    In this case, the flow is~$f((u,i))=\pi(u)/H^{t+1}$, so~$z_i(u)=1$, which implies that~$z_{j,i}=1$ and~$i\in Z_j$ for all~$j$.
    Hence, every subset~$S$ in the support of~~$y^{t+1}(V^t,A^t,v^{t+1},\cdot)$ contains~$i$, as desired.

    It remains to show that~$y^{t+1}(V^t,A^t,v^{t+1},\cdot)$ satisfies Equation~\eqref{eq:EAPdef} for step~$t+1$.
    Recall that~$\calA^t$ is the set of all cumulative allocations~$A^t$ that occur with strictly positive probability.
    We need to show that the probability that party~$i$ is allocated a seat in step~$t+1$ equals~$v_i^{t+1}$.
    To compute this, observe that
    \begin{align*}
        \sum_{A^{t}\in\calA^{t}} \pi(y^{t},V^{t},A^{t}) \cdot a^{t+1}_i(V^{t},A^{t},v^{t+1})
        &=\sum_{u\in U^t} \sum_{A^t\in \calA^t:u=u(A^t)} \pi(y^{t},V^{t},A^{t}) \cdot a^{t+1}_i(V^{t},A^{t},v^{t+1}) \\
        &=\sum_{u\in U^t} \pi(u) \cdot z_i(u) \\
        &=\sum_{u\in U^t} \pi(u) \cdot \frac{f(u,i)\cdot H^{t+1}}{\pi(u)}\\
        &=\sum_{u\in U^t} f(u,i) \cdot H^{t+1}.
    \end{align*}
    By flow conservation at vertex~$i$, the total incoming flow equals the outgoing flow~$f((i,d))=v_i^{t+1}/H^{t+1}$, so the final expression evaluates to~$v_i^{t+1}$, as required.
\end{proof}

\subsection{Proof of \Cref{lem:flowpolytope}\label{app:lem:flowpolytope}}

\lemflowpolytope*

\begin{proof}
    Since the case~$n=2$ follows immediately from~$n=3$ by setting~$v_3^t=0$ for all~$t$, we focus on~$n=3$.
    For~$t=1$, the construction of~$y^1$ is precisely given in the construction of the partial network flow method.
    Thus, let~$t\geq 1$ and let a partial network flow method~$(y^k)_{k\in[t]}$ up to step~$t$ be given.
    We show that the corresponding flow network~$\text{FN}(V^t,A^t,v^{t+1},\pi)$ allows for a feasible flow of value~$1$ using a detailed case distinction based on
    \begin{itemize}
        \item the number of seats to be allocated in step~$t+1$ (i.e., whether~$H^{t+1}=1$ or~$H^{t+1}=2$),
        \item the number of parties that have been assigned their upper share of seats after step~$t$ (i.e., whether~$\vert u\vert=1$ or~$\vert u\vert=2$ for every~$u\in U^t$), and
        \item how many parties obtained enough votes in step~$t+1$ to increase their upper share of seats by~$1$.
    \end{itemize}

    This results in eight distinct cases, listed below.
    For each, we show that the corresponding flow network admits a feasible flow of value~$1$.
    We begin by making two general observations, using the fact that the partial network flow method~$(y^k)_{k\in[t]}$ satisfies Equation~\eqref{eq:EAPdef} for every~$k\in [t]$.
    For any party~$i\in [3]$ and set~$u\in U^t$, we have
    \begin{equation}\label{eq:P[u]}
    \begin{aligned}
        \pi(u)&=\sum_{k=1}^t v_i^k-\left\lfloor \sum_{k=1}^t v_i^k\right \rfloor
        \leq 1- \left( \left\lceil \sum_{k=1}^t v_i^k \right \rceil -\sum_{k=1}^t v_i^k \right) 
        &&\text{ if } u=\{i\}, \text{ and } \\
        \pi(u)&=1-\left( \sum_{k=1}^t v_i^k - \left\lfloor \sum_{k=1}^t v_i^k \right \rfloor \right)
        \geq \left\lceil \sum_{k=1}^t v_i^k \right \rceil -\sum_{k=1}^t v_i^k 
        &&\text{ if } i\notin u,~\vert u\vert =2.
    \end{aligned}
    \end{equation}
    The inequalities are strict if and only if~$\sum_{k=1}^t v_i^k\in \NN$.
    Furthermore, in any flow network, the vertex sets~$P_o=\{o\}$ and~$P_d=P\setminus\{d\}$ each induce a cut of capacity~$1$.
    Thus, any flow of value~$1$ must assign any edge adjacent to~$o$ or~$d$ a flow equal to the edge's upper capacity.

    The remainder of this proof is structured as follows:
    We first list the eight cases, then address each one separately, showing the existence of a feasible flow~$f$ of value~$1$.
    All relevant networks are depicted in \Cref{fig:case1,fig:case2}, where edges~$(u,i)$ with~$c((u,i))=0$ are omitted for clarity.

    \begin{enumerate}[label=\arabic*.]
        \item $H^{t+1}=1$
        \begin{enumerate}[label=\arabic{enumi}.\arabic*.]
            \item $\vert u \vert=1$ for all~$u\in U^t$
            \begin{enumerate}[label=\arabic{enumi}.\arabic{enumii}.\arabic*.]
                \item $\forall i\in [3] : v_i^{t+1}<\left\lceil \sum_{k=1}^t v^k_i\right\rceil - \sum_{k=1}^t v^k_i$
                \item $\exists !~i\in [3]: v_i^{t+1}\geq \left\lceil \sum_{k=1}^t v^k_i\right\rceil - \sum_{k=1}^t v^k_i>0$
            \end{enumerate}
            \item $\vert u \vert=2$ for all~$u\in U^t$
            \begin{enumerate}[label=\arabic{enumi}.\arabic{enumii}.\arabic*.]
                \item $\exists !~i\in [3]: v_i^{t+1}\geq \left\lceil \sum_{k=1}^t v^k_i\right\rceil - \sum_{k=1}^t v^k_i>0$
                \item $\exists !~i\in [3]: v_i^{t+1}< \left\lceil \sum_{k=1}^t v^k_i\right\rceil - \sum_{k=1}^t v^k_i$
            \end{enumerate}
        \end{enumerate}
        \item $H^{t+1}=2$
        \begin{enumerate}[label=\arabic{enumi}.\arabic*.]
            \item $\vert u \vert=1$ for all~$u\in U^t$
            \begin{enumerate}[label=\arabic{enumi}.\arabic{enumii}.\arabic*.]
                \item $\exists !~i\in [3]: v_i^{t+1}\geq \left\lceil \sum_{k=1}^t v^k_i\right\rceil - \sum_{k=1}^t v^k_i>0$
                \item $\exists !~i\in [3]: v_i^{t+1}< \left\lceil \sum_{k=1}^t v^k_i\right\rceil - \sum_{k=1}^t v^k_i$
            \end{enumerate}
            \item $\vert u \vert=2$ for all~$u\in U^t$
            \begin{enumerate}[label=\arabic{enumi}.\arabic{enumii}.\arabic*.]
                \item $\exists !~i\in [3]: v_i^{t+1}< \left\lceil \sum_{k=1}^t v^k_i\right\rceil - \sum_{k=1}^t v^k_i$
                \item $\forall i\in [3] : v_i^{t+1}\geq \left\lceil \sum_{k=1}^t v^k_i\right\rceil - \sum_{k=1}^t v^k_i>0$
            \end{enumerate}
        \end{enumerate}
     \end{enumerate}
     Note that this covers all possible cases due to the following observation:
     For each party~$i\in[n]$, we have
     \[
     \sum_{k=1}^t v^k_i - \left\lfloor \sum_{k=1}^t v^k_i\right\rfloor + v^{t+1}_i \in [0,2),
     \]
     while the sum evaluates to
     \[
     \sum_{i\in[3]} \left(\sum_{k=1}^t v^k_i - \left\lfloor \sum_{k=1}^t v^k_i\right\rfloor + v^{t+1}_i\right) = |u|+H^{t+1}.
     \]
     This implies that the number of parties for which the individual term is at least~$1$ lies within~$H^{t+1}+\vert u\vert-2$ and~$H^{t+1}+\vert u\vert-1$.

    \input{fig_case1}
     
    \paragraph{Case~1.1.1.} 
    The flow network is illustrated in \Cref{fig:case1.1.1}.
    In this case, the flow~$f$ of value~$1$ is not unique.
    However, any feasible flow~$f$ with value~$1$ satisfies~$f(e)=c(e)$ for all
    \[
        e\in \{(o,\{i\}),(\{i\},i),(i,d): i\in [3]\}.
    \]
    The flow on the remaining edges must satisfy the following set of linear equalities
    \begin{align*}
        f((\{1\},2)+f((\{1\},3)&=\pi(\{1\}), \\
        f((\{2\},1)+f((\{2\},3)&=\pi(\{2\}), \\
        f((\{3\},1)+f((\{3\},2)&=\pi(\{3\}), \\
        f((\{2\},1)+f((\{3\},1)&=v_1^{t+1}, \\
        f((\{1\},2)+f((\{3\},2)&=v_2^{t+1}, \\
        f((\{1\},3)+f((\{2\},3)&=v_3^{t+1}, \\
        f(e)\geq 0.
    \end{align*}
    Since adding the first three equalities yields the same as adding the other three, the set of linear equalities is underdetermined, i.e., there is no unique solution.
    We show that the system has a feasible solution.
    
    To this end, we interpret the inner part of the network as a bipartite graph~$G\coloneqq(P_1\cup P_2, E)$ with vertex sets~$P_1\coloneqq \{\{1\},\{2\},\{3\}\}$ and~$P_2\coloneqq \{1,2,3\}$, and edge set~$E\coloneqq\{(\{i\},j) : i\neq j\}$.
    Each vertex~$\{i\}\in P_1$ has supply~$b(\{i\})\coloneqq\pi(\{i\})$ and each vertex~$i\in P_2$ has demand~$b(i)\coloneqq v_i^{t+1}$ for every~$i\in [3]$.
    Thus, each edge is adjacent to one supply and one demand vertex.
    Since edge capacities match the supply of the adjacent supply vertex, we can treat this as a transshipment instance without explicit edge capacities.

    We claim that this transshipment instance has a feasible solution.
    First, note that the total supply and demand coincide.
    Now it suffices to verify that for all~$P'\subseteq P_1$, the total supply~$\sum_{p\in P'}b(p)$ is at most the total demand~$\sum_{p\in N(P')} b(p)$ of its neighborhood~$N(P')$ (cf. \cite[Corollary 11.2.g]{schrijver2004combinatorial}).
    We focus on three cases.
    For~$P'=\{\{1\}\}$, we have
    \[
    \sum_{p\in P'}b(p)
    =\pi(\{1\})
    \leq 1-\left( \left\lceil \sum_{k=1}^t v_1^k \right \rceil -\sum_{k=1}^t v_1^k \right)
    \leq 1- v_1^{t+1}
    =v_2^{t+1}+v_3^{t+1}
    =\sum_{v\in N(P')} b(p).
    \]
    For~$P'=\{\{1\},\{2\}\}$, we obtain
    \[
    \sum_{p\in P'}b(p)
    =\pi(\{1\})+\pi(\{2\})
    \leq 1
    =v_1^{t+1}+v_2^{t+1}+v_3^{t+1}
    =\sum_{p\in N(P')} b(p),
    \]
    and finally for~$P'=\{\{1\},\{2\},\{3\}\}$, we have
    \[
    \sum_{p\in P'}b(p)
    =\pi(\{1\})+\pi(\{2\})+\pi(\{3\})
    =1
    =v_1^{t+1}+v_2^{t+1}+v_3^{t+1}
    =\sum_{p\in N(P')} b(p).
    \]
    By symmetry, the condition also holds for all other subsets~$P'$.
    Hence, a feasible transshipment exists, implying the existence of a feasible flow~$f$ of value~$1$ in the original network.
    
    \paragraph{Case~1.1.2.}
    There exists exactly one party~$i\in[3]$ with~$v_i^{t+1}\geq \lceil \sum_{k=1}^t v^k_i\rceil - \sum_{k=1}^t v^k_i$.
    Without loss of generality, assume~$i=1$.
    The flow network is illustrated in \Cref{fig:case1.1.2}.
    We define~$f$ as follows.
    \begin{align*}
        f(e)=
        \begin{cases}
            v_1^{t+1}-(1-\pi(\{1\})) & \text{ if } e=(\{1\},1),\\
            v_2^{t+1} & \text{ if } e=(\{1\},2),\\
            v_3^{t+1} & \text{ if } e=(\{1\},3),\\
            0 & \text{ if } e=(\{2\},3) \text{ or } e=(\{3\},2),\text{ and } \\
            c(e) & \text{ otherwise.}
        \end{cases}
    \end{align*}
    We now verify feasibility.   
    For~$e=(\{1\},1)$, since~$1-\pi(\{1\}) \leq v_1^{t+1} < 1$ we have~$f(e)\in [0,\pi(\{1\})]$.
    For~$e=(\{1\},j)$ with~$j\in\{2,3\}$, we have~$v_j^{t+1}\leq \pi(\{1\})$ since~$v_1^{t+1}+v_j^{t+1}\leq 1 \leq v_1^{t+1}+\pi(\{1\})$.
    Thus, all capacities are respected.
    Flow conservation at vertex $1$ holds because
    \[
    \sum_{u\in U^t} f((u,1))=v_1^{t+1}-(1-\pi(\{1\})) + \pi(\{2\})+\pi(\{3\}) =v_1^{t+1} = f((1,d)).
    \]
    Flow conservation at all other inner vertices is trivial, and since the outgoing flow from~$o$ is~$1$,~$f$ is a feasible flow of value~$1$.

    \paragraph{Case~1.2.1.}
    There exists exactly one party~$i\in[3]$ with~$v_i^{t+1}\geq \lceil \sum_{k=1}^t v^k_i\rceil - \sum_{k=1}^t v^k_i$.
    Without loss of generality, assume~$i=1$.
    The flow network is illustrated in \Cref{fig:case1.2.1}.
    We define~$f$ as follows.
    \begin{align*}
        f(e)=
        \begin{cases}
            \pi(\{1,2\})-v_3^{t+1} & \text{ if } e=(\{1,2\},1),\\
            v_3^{t+1} & \text{ if } e=(\{1,2\},3), \\
            \pi(\{1,3\})-v_2^{t+1} & \text{ if } e=(\{1,3\},1),\\
            v_2^{t+1} & \text{ if } e=(\{1,3\},2), \\
            c(e) & \text{ otherwise.}
        \end{cases}
    \end{align*}
    To verify feasibility, observe that~$v_2^{t+1}<\pi(\{1,3\})$ and~$v_3^{t+1}<\pi(\{1,2\})$ as implied by Property~\eqref{eq:P[u]}.
    Hence, all flow values remain within the capacity bounds.
    Flow conservation at vertex~$1$ holds because
    \[
    \sum_{u\in U^t} f((u,1))=\pi(\{1,2\})-v_3^{t+1}+\pi(\{1,3\})-v_2^{t+1}+\pi(\{2,3\})=1-v_2^{t+1}-v_3^{t+1}=v_1^{t+1}=f((1,d)).
    \]
    Flow conservation at all other inner vertices is trivial, and since the outgoing flow from~$o$ is~$1$,~$f$ is a feasible flow of value~$1$.

    \paragraph{Case~1.2.2.}
    There exists exactly one party~$i\in[3]$ with~$v_i^{t+1}< \lceil \sum_{k=1}^t v^k_i\rceil - \sum_{k=1}^t v^k_i$.
    Without loss of generality, assume~$i=1$.
    The flow network is illustrated in \Cref{fig:case1.2.2}.
    We define~$f$ as follows.
    \begin{align*}
        f(e)=
        \begin{cases}
            v_1^{t+1}& \text{ if } e=(\{2,3\},1),\\
            v_2^{t+1}-\pi(\{1,3\}) & \text{ if } e=(\{2,3\},2),\\
            v_3^{t+1}-\pi(\{1,2\}) & \text{ if } e=(\{2,3\},3),\\
            0 & \text{ if } e=(\{1,2\},2) \text{ or } e=(\{1,3\},3),\text{ and } \\
            c(e) & \text{ otherwise.}
        \end{cases}
    \end{align*}
    Flow conservation at vertex~$\{2,3\}$ follows from
    \[
    \sum_{j=1}^3 f((\{2,3\},j))
    =v_1^{t+1}+v_2^{t+1}+v_3^{t+1}-\pi(\{1,3\})-\pi(\{1,2\})\\
    =\pi(\{2,3\})
    =f((o,\{2,3\})).
    \]
    Flow conservation at all other inner vertices is trivial.
    The non-negativity of~$f$ follows from Property~\eqref{eq:P[u]}, which is satisfied with equality in this case since parties~$j\in \{2,3\}$ satisfy $v_j^{t+1}\geq \lceil \sum_{k=1}^t v^k_j\rceil - \sum_{k=1}^t v^k_j$ and therefore $\sum_{k=1}^t v^k_j\notin \NN$.
    With~$\sum_{j=1}^3 f((\{2,3\},j))=\pi(\{2,3\})$, the non-negativity implies~$f(e)\leq \pi(\{2,3\})$ for all~$e=(\{2,3\},j)$.
    Since the outgoing flow from~$o$ is~$1$,~$f$ is a feasible flow of value~$1$.

    \input{fig_case2}
     
    \paragraph{Case~2.1.1.}  
    There exists exactly one party~$i\in[3]$ with~$v_i^{t+1}\geq \lceil \sum_{k=1}^t v^k_i\rceil - \sum_{k=1}^t v^k_i$.
    Without loss of generality, assume~$i=1$.
    The flow network is illustrated in \Cref{fig:case2.1.1}.
    We define~$f$ as follows.
    \begin{align*}
        f(e)=
        \begin{cases}
            (v_1^{t+1}-(1-\pi(\{1\})))/2 & \text{ if } e=(\{1\},1),\\
            (v_2^{t+1}-\pi(\{3\}))/2 & \text{ if } e=(\{1\},2),\\
            (v_3^{t+1}-\pi(\{2\}))/2 & \text{ if } e=(\{1\},3), \text{ and }\\
            c(e) & \text{ otherwise.}
        \end{cases}
    \end{align*}
    First we show flow conservation at vertex~$\{1\}$ by
    \begin{align*}
    \sum_{j=1}^3 f((\{1\},j))
    =\frac{v_1^{t+1}+v_2^{t+1}+v_3^{t+1}-\pi(\{2\})-\pi(\{3\})-(1-\pi(\{1\}))}{2}
    =\pi(\{1\})
    =f((o,\{1\})).
    \end{align*}
    Flow conservation at vertex $1$ holds because 
    \[
        \sum_{u\in U^t}f((u,1)) = \frac{v_1^{t+1}-(1-\pi(\{1\}))}{2} + \frac{\pi(\{2\})}{2}+\frac{\pi(\{3\})}{2} = \frac{v^{t+1}_1}{2} = f((1,d)).
    \]
    Flow conservation at all other inner vertices is trivially satisfied.
    
    To check the capacity constraints, consider~$e=(\{1\},1)$.
    Since~$v_1^{t+1}\geq 1-\pi(\{1\})$, we have~$f(e)\geq 0$.
    For~$e=(\{1\},2)$, note that
    \[
    v_2^{t+1}
    =2-v_1^{t+1}-v_3^{t+1}
    \geq 2-v_1^{t+1}-\left( \left\lceil \sum_{k=1}^t v^k_3\right\rceil - \sum_{k=1}^t v^k_3\right)
    \geq 1-v_1^{t+1}+\pi(\{3\})
    >\pi(\{3\}),
    \]where the first inequality follows from the case assumption, the second inequality from Property~\eqref{eq:P[u]}, and the third (strict) inequality from~$v_1^{t+1}<1$.
    Hence,~$f(e)>0$ and similar arguments yield~$f(e)>0$ for~$e=(\{1\},3)$.

    Next, we verify upper capacity constraints.
    For~$e=(\{1\},1)$, since~$v_1^{t+1}<1$, we have~$f(e)<\pi(\{1\})/2$.
    For~$e=(\{1\},2)$,
    \[
    f(e)=\frac{v_2^{t+2}-\pi(\{3\})}{2} \leq \frac{1-\pi(\{2\})-\pi(\{3\})}{2}=\frac{\pi(\{1\})}{2}=c(e),
    \]
    and analogously for~$e=(\{1\},3)$.
    Since the outgoing flow from~$o$ is~$1$,~$f$ is a feasible flow of value~$1$.

    \paragraph{Case~2.1.2.}
    There exists exactly one party~$i\in[3]$ with~$v_i^{t+1}< \lceil \sum_{k=1}^t v^k_i\rceil - \sum_{k=1}^t v^k_i$.
    Without loss of generality, assume~$i=1$.
    The flow network is illustrated in \Cref{fig:case2.1.2}.
    We define~$f$ as follows.
    \begin{align*}
        f(e)=
        \begin{cases}
            (1-v_2^{t+1})/2 & \text{ if } e=(\{2\},1),\\
            (v_2^{t+1}-(1-\pi(\{2\})))/2 & \text{ if } e=(\{2\},2),\\
            (1-v_3^{t+1})/2 & \text{ if } e=(\{3\},1),\\
            (v_3^{t+1}-(1-\pi(\{3\})))/2 & \text{ if } e=(\{3\},3), \text{ and }\\
            c(e) & \text{ otherwise.}
        \end{cases}
    \end{align*}
    We first verify flow conservation.
    Consider vertex~$\{2\}$ and observe that
    \[
    \sum_{j=1}^3 f((\{2\},j))
    =\frac{1-v_2^{t+1}+v_2^{t+1}-(1-\pi(\{2\}))+\pi(\{2\})}{2}    
    =\pi(\{2\})
    =f((o,\{2\})).
    \]
    An analogous computation applies to~$\{3\}$.
    For vertex~$1$, flow conservation holds because
    \[
    \sum_{j=1}^3 f((\{j\},1)
    =\frac{1-v_2^{t+1}+1-v_3^{t+1}}{2}
    =\frac{v_1^{t+1}}{2}
    =f((3,d)).
    \]
    For the remaining inner vertices, flow conservation is trivially satisfied.

    Next, we show that~$f$ obeys the capacity constraints.
    With~$1-\pi(\{j\})=\lceil\sum_{k=1}^t v^k_j\rceil - \sum_{k=1}^t v^k_j \leq v_j^{t+1}<1$ for~$j\in\{2,3\}$, all lower capacity constraints are satisfied.
    For~$e=(\{2\},1)$, observe that
    \[
    f(e)=\frac{1-v_2^{t+1}}{2}
    \leq \frac{1-(1-\pi(\{2\}))}{2}
    = \frac{\pi(\{2\})}{2}
    =c(e).
    \]
    An analogous computation yields~$f(e)\leq c(e)$ for~$e=(\{3\},1)$.
    For edges~$e=(\{2\},2)$ and~$e=(\{3\},3)$ the upper capacities are trivially satisfied since~$v_j^{t+1}<1$ for all~$j$.
    Since the outgoing flow from~$o$ is~$1$,~$f$ is a feasible flow of value~$1$.

    \paragraph{Case~2.2.1.}
    There exists exactly one party~$i\in[3]$ with~$v_i^{t+1}< \lceil \sum_{k=1}^t v^k_i\rceil - \sum_{k=1}^t v^k_i$.
    Without loss of generality, assume~$i=1$.
    The flow network is illustrated in \Cref{fig:case2.2.1}.
    We define~$f$ as follows.
    \begin{align*}
        f(e)=
        \begin{cases}
            v_1^{t+1}/2 & \text{ if } e=(\{2,3\},1),\\
            (v_2^{t+1}-(1-\pi(\{2,3\})))/2 & \text{ if } e=(\{2,3\},2),\\
            (v_3^{t+1}-(1-\pi(\{2,3\})))/2 & \text{ if } e=(\{2,3\},3), \text{ and }\\
            c(e) & \text{ otherwise.}
        \end{cases}
    \end{align*}
    We first verify flow conservation.
    For vertex~$\{2,3\}$ observe that
    \begin{align*}
    \sum_{j=1}^3 f((\{2,3\},j))
    &=\frac{v_1^{t+1}+v_2^{t+1}-(1-\pi(\{2,3\}))+v_3^{t+1}-(1-\pi(\{2,3\}))}{2}\\
    &=\frac{\pi(\{2,3\})}{2}
    =f((o,\{2,3\})).
    \end{align*}
    For the remaining inner vertices, flow conservation is trivially satisfied.

    Next, we show that~$f$ obeys the capacity constraints.
    For~$e=(\{2,3\},2)$, note that
    \[
    v_2^{t+1}
    =2-v_1^{t+1}-v_3^{t+1}
    \geq 2-v_3^{t+1}-\left( \left\lceil \sum_{k=1}^t v^k_1\right\rceil - \sum_{k=1}^t v^k_1 \right)
    \geq 2-v_3^{t+1}-\pi(\{2,3\})
    >1-\pi(\{2,3\}),
    \]
    where the first inequality follows from the case assumption, the second inequality from Property~\eqref{eq:P[u]}, and the third (strict) inequality from~$v_3^{t+1}<1$.
    An analogous computation applies to~$e=(\{2,3\},3)$, proving that~$\ell(e)\leq f(e)$ holds for all edges.
    For the upper capacity constraints, consider~$e=(\{2,3\},1)$.
    Property~\eqref{eq:P[u]} yields
    \[
    v_1^{t+1}
    <\left\lceil \sum_{k=1}^t v^k_1\right\rceil - \sum_{k=1}^t v^k_1
    \leq \pi(\{2,3\}),
    \]
    i.e.,~$f(e)<c(e)$.
    For~$e=(\{2,3\},2)$, we have~$(v_2^{t+1}-(1-\pi(\{2,3\}))/2<\pi(\{2,3\})/2$ since~$v_2^{t+1}<1$, and analogously for~$e=(\{2,3\},3)$.
    Since the outgoing flow from~$o$ is~$1$,~$f$ is a feasible flow of value~$1$.

    \paragraph{Case~2.2.2.}    
    The flow network is illustrated in \Cref{fig:case2.2.2}.
    This case is structurally similar to Case~1.1.1: the flow~$f$ is not uniquely determined.
    We begin by fixing the flow along the edges~$(\{1,2\},3),(\{1,3\},2)$, and $(\{2,3\},1)$, each of which lies on a unique~$(o,d)$-path.
    These edges must carry their respective lower bounds, which are~$\pi(\{1,2\})/2$,~$\pi(\{1,3\})/2$, and~$\pi(\{2,3\})/2$.

    Fixing these values reduces the residual problem to the flow network depicted in~\Cref{fig:case3}, where all remaining lower bounds are zero.
    The goal is now to route a residual flow of value
    \[
    1-\frac{\pi(\{1,2\}+\pi(\{1,3\})+\pi(\{2,3\})}{2}=\frac{1}{2}.
    \]

    To establish feasibility, we again reduce to a transshipment instance on a bipartite graph~$G\coloneqq(P_1\cup P_2, E)$ with vertex sets~$P_1\coloneqq \{\{1,2\},\{1,3\},\{2,3\}\}$, and~$P_2\coloneqq \{1,2,3\}$, and edge set~$E\coloneqq\{(\{i\},j) : i\neq j\}$.
    Each vertex~$\{i,j\}\in P_1$ has supply~$b(\{i,j\})\coloneqq\pi(\{i,j\})/2$ and each vertex~$i\in P_2$ has demand~$b(i)\coloneqq(v_i^{t+1}-\pi(\{j,k\}))/2$ where~$\{j,k\}=[3]\setminus\{i\}$.
    Thus, each edge is adjacent to one supply and one demand vertex.
    Since edge capacities match the supply of the adjacent supply vertex, we can treat this as a transshipment instance without explicit edge capacities.

    We claim that this transshipment instance has a feasible solution.
    First, note that the total supply and demand coincide.
    Now it suffices to verify that for all~$P'\subseteq P_1$, the total supply~$\sum_{p\in P'}b(p)$ is at most the total demand~$\sum_{p\in N(P')} b(p)$ of its neighborhood~$N(P')$.
    We focus on three cases.
    For~$P'=\{\{1,2\}\}$, we have
    \begin{align*}
        \sum_{p\in P'}b(p)
        &=\frac{\pi(\{1,2\}}{2}
        =\frac{1-\pi(\{1,3\})-\pi(\{2,3\})}{2}\\
        &=\frac{v_1^{t+1}+v_2^{t+1}+v_3^{t+1}-1-\pi(\{1,3\})-\pi(\{2,3\})}{2}\\
        &<\frac{v_1^{t+1}+v_2^{t+1}-\pi(\{1,3\})-\pi(\{2,3\})}{2}
        =\sum_{p\in N(P')} b(p).
    \end{align*}
    For~$P'=\{\{1,2\},\{1,3\}\}$, we obtain
    \[
    \sum_{p\in P'}b(p)
    =\frac{\pi(\{1,2\})+\pi(\{1,3\})}{2}
    \leq \frac{1}{2}
    =\sum_{p\in N(P')} b(p).
    \]
    Finally, for~$P'=\{\{1,2\},\{1,3\},\{2,3\}\}$, we have
    \[
    \sum_{p\in P'}b(p)
    =\frac{\pi(\{1,2\})+\pi(\{1,3\})+\pi(\{2,3\})}{2}
    =\frac{1}{2}
    =\sum_{p\in N(P')} b(p).
    \]
    By symmetry, the condition also holds for all other subsets~$P'$.
    Thus, the transshipment instance is feasible, which implies the existence of a feasible flow of value~$1/2$ in the residual network.
    Together with the fixed lower bound flow, we obtain a feasible flow of value~$1$, completing this case.
    
    \input{fig_case3}
    This completes the case distinction.    
\end{proof}

\subsection{Proof of \Cref{prop:methods-one-to-one}\label{app:prop:methods-one-to-one}}

\lemmethodsonetoone*

\begin{proof}
    Consider an online apportionment method~$(y^t)_{t\in\NN}$ satisfying global quota and ex-ante proportionality.
    We need to show that one can obtain the same method using the network flow construction.
    First, consider step~$t=1$.
    Since~$\calA^0=\{0\}$ and~$V^0=0$, there is just one probability distribution~$y^1$ over the subsets of~$[n]$ of size~$H^1$.
    Let~$Z_1,\dots,Z_m\subset [n]$ be the support of~$y^1$ and let~$\lambda_1,\dots,\lambda_m\in[0,1]$ be the corresponding probabilities, i.e.,~$y^1(V^0,A^0,v^1,Z_j)=\lambda_j$ for every~$j$.
    Let~$z_j\in\{0,1\}^n$ denote the characteristic vector of~$Z_j$, i.e.,~$z_{j,i}=1$ if and only if~$i\in Z_j$.
    Since~$\vert Z_j \vert =H^1$ for every~$j$, the vector~$z_j$ has exactly~$H^1$ ones.
    Furthermore, ex-ante proportionality yields~$\sum_{j=1}^m \lambda_j\cdot z_{j}=v^1$.
    Thus, the probability distribution~$y^1$ yields a unique convex decomposition into~$n$-dimensional~$\{0,1\}$-vectors with exactly~$H^1$ ones.
    This convex decomposition can be chosen during the construction of a partial network flow method.

    Next, assume that up to step~$t$, the online apportionment method~$(y^k)_{k\in\NN}$ coincides with a partial network flow method~$(z^k)_{k\in [t]}$.
    Let~$\text{FN}(V^t,A^t,v^{t+1},\pi)$ be the flow network corresponding to the partial network flow method~$(z^k)_{k\in [t]}$ at step~$t+1$.
    We show how the method~$(y^k)_{k\in\NN}$ for step~$t+1$ yields a flow of value~$1$ in~$\text{FN}(V^t,A^t,v^{t+1},\pi)$.

    For each~$A^t\in\calA^t$, the method~$(y^t)_{t\in\NN}$ determines the probability~$\pi(y^t,V^t,A^t)$ with which the cumulative allocation~$A^t$ is attained.
    This uniquely determines the probability~$\pi(u(A^t))$ that exactly the parties in~$u(A^t)$ have been allocated their upper share of seats.
    Let~$\calA^t=\calA^t_1\cup\dots\cup\calA^t_m$ be a partition of~$\calA^t$ into disjoint subsets such that~$u(A^t)=u(\tilde{A}^t)$ if and only if~$A^t,\tilde{A}^t\in\calA^t_j$ for some~$j\in[m]$.
    Note that the flow network has exactly one vertex~$u_j$ for each~$\calA_j^t$.
    The probability that method~$(y^t)_{t\in\NN}$ assigns a seat to party~$i$ in step~$t+1$ conditioned on the cumulative allocation~$u_j$ is denoted by
    \[
    \pi_{u_j}(i)=\sum_{A^t\in \calA^t_j} \pi(y^t,V^t,A^t) \cdot \sum_{S\subseteq[n]:i\in S} y^{t+1}(V^t,A^t,v^{t+1},S).
    \]
    
    We define~$f$ as follows.
    For each edge~$(o,u_j)$ from the origin to a vertex~$u_j$, set~$f((o,u_j))=\pi(u_j)$;
    for each edge~$(u_j,i)$ from set~$u_j$ to party~$i$, set~$f((u_j,i))=\pi_{u_j}(i)/H^{t+1}$; and
    finally, for each edge~$(i,d)$ from a vertex~$i$ to the destination, set~$f((i,d))=v_i^{t+1}/H^{t+1}$.
    We claim that this defines a feasible flow of value~$1$.

    First, we show that~$f$ obeys the edge capacities.
    This is immediate for edges~$(o,u^t)$ and~$(i,d)$ as~$f$ equals the upper capacities.
    To verify feasibility on edges~$(u_j,i)$, we distinguish the three cases given by the flow construction.
    First, assume that~$i\in u_j$ and that the upper share of seats does not increase from step~$t$ to~$t+1$, i.e.,~$\lceil \sum_{k=1}^t v_i^k\rceil=\lceil \sum_{k=1}^{t+1} v_i^k\rceil$ or~$v_i^{t+1}=0$.
    We need to show that~$f((u_j,i))=0$.
    However, in that case, party~$i$ cannot be assigned a seat in step~$t+1$ without violating local or global quota.
    Since the method~$(y^t)_{t\in\NN}$ fulfills these properties, it assigns probability~$0$ to any set containing~$i$ for every~$A^t\in\calA^t_j$.
    Thus, we have~$\pi_{u_j}(i)=0$, implying~$f((u_j,i))=\pi_{u_j}(i)/H^{t+1}=0$.
    Next, assume~$i\notin u_j$ and the lower share increases from step~$t$ to~$t+1$, i.e.,~$\lfloor \sum_{k=1}^t v_i^k\rfloor+1=\lfloor \sum_{k=1}^{t+1} v_i^k\rfloor$.
    We need to show that~$f((u_j,i))=\pi(u_j)/H^{t+1}$.
    However, in that case, party~$i$ must receive a seat in step~$t+1$, as otherwise global quota would be violated.
    Since the method~$(y^t)_{t\in\NN}$ fulfills these properties, it only assigns positive probability to sets containing~$i$, which yields~$\sum_{S\subseteq[n]:i\in S}y^{t+1}(V^t,A^t,v^{t+1},S)=1$ for every~$A^t\in\calA^t_j$.
    Hence, we obtain
    \[
    f((u_j,i))=\frac{\pi_{u_j}(i)}{H^{t+1}}= \sum_{A\in \calA^t_j} \pi(y^t,V^t,A^t) \cdot \frac{1}{H^{t+1}}=\frac{\pi(u_j)}{H^{t+1}}.
    \]
    as desired.
    In all other cases, the method~$(y^t)_{t\in\NN}$ is less restricted and we have~$y^{t+1}(V^t,A^t,v^{t+1},S)\in[0,1]$ for any subset~$S\subseteq [n]$.
    We only need to show that~$f((u_j,i))\in[0,\pi(u_j)/H^{t+1}]$.
    This follows from
    \[
    f((u_j,i))=\frac{\pi_{u_j}(i)}{H^{t+1}}\leq \sum_{A\in \calA^t_j} \pi(y^t,V^t,A^t) \cdot \frac{1}{H^{t+1}}=\frac{\pi(u_j)}{H^{t+1}}.
    \]
    Thus,~$f$ obeys the capacity constraints.
   
    To verify flow conservation, we first consider an inner vertex~$u_j$.
    The incoming flow from~$o$ is~$f((o,u_j))=\pi(u_j)$.
    The outgoing flow is
    \begin{align*}
    \sum_{i=1}^n f((u_j,i))&=\sum_{i=1}^n \frac{\pi_{u_j}(i)}{H^{t+1}}\\
    &=
    \sum_{A\in \calA^t_j} \pi(y^t,V^t,A^t) \cdot \frac{1}{H^{t+1}}\cdot \left(\sum_{i=1}^n \sum_{S\subseteq[n]:i\in S} y^{t+1}(V^t,A^t,v^{t+1},S)\right)
    =\pi(u_j),
    \end{align*}
    where the last equality holds, since the method~$(y^t)_{t\in\NN}$ assigns exactly~$H^{t+1}$ seats in step~$t+1$.
    For an inner vertex~$i$,~$(y^t)_{t\in\NN}$ satisfying ex-ante proportionality yields
    \begin{align*}
    \sum_{u_j\in U^t} f((u_j,i))&= \sum_{u_j\in U^t}\frac{\pi_{u_j}(i)}{H^{t+1}}\\
    &=\sum_{u_j\in U^t} \frac{1}{H^{t+1}} \sum_{A\in \calA^t_j} \pi(y^t,V^t,A^t) \cdot \sum_{S\subseteq[n]:i\in S} y^{t+1}(V^t,A^t,v^{t+1},S)
    =\frac{v^{t+1}}{H^{t+1}}.
    \end{align*}
    Thus, the incoming flow matches the outgoing flow~$f((i,d))$, so flow conservation holds at every inner vertex.

    Since~$f$ obeys the capacity constraints of the flow network~$\text{FN}(V^t,A^t,v^{t+1},\pi)$ and attains flow conservation, it is a feasible flow.    
    Finally, the total flow leaving the origin~$o$ is~$\sum_{u_j\in U^t}\pi(u_j)=1$.
    Hence, the flow~$f$ has value~$1$ as claimed.
    This exact flow allows to extend the partial network flow method~$(z^k)_{k\in [t]}$ up to step~$t+1$ such that the online apportionment method~$(y^k)_{k\in\NN}$ coincides with the partial network flow method~$(z^k)_{k\in [t+1]}$.
    This finally proves that the online apportionment method~$(y^k)_{k\in\NN}$ is a network flow method.
\end{proof}

\newpage
\bibliographystyle{abbrvnat} 
\bibliography{literatur.bib}

\end{document}

%% file: offlineflownetwork.tex
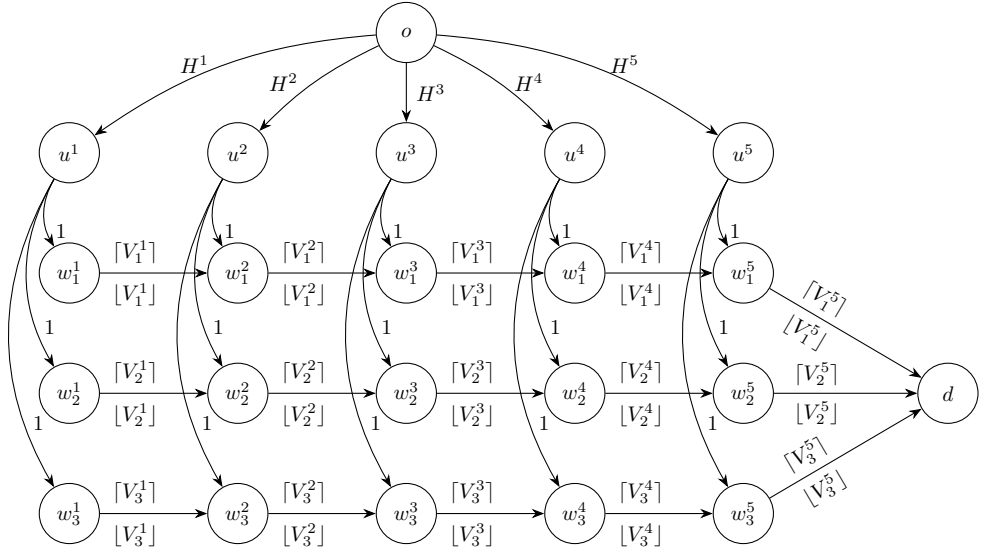
\begin{figure}[t]
\centering
\begin{tikzpicture}[
    node distance=1.5cm and 2cm,
    vertex/.style={draw, circle, minimum size=1cm, font=\small},
    ->, >=Stealth,
    scale=0.8, transform shape
]

    \foreach \t [evaluate=\t as \x using (\t-3)*2.8] in {1,...,5} {
        \node[vertex] (u\t) at (\x, 0) {$u^{\t}$};
    }

    \node[vertex] (o) at (0, 2) {$o$};
    \draw (o) to[bend right=15] node[left=5pt, font=\small] {$H^1$} (u1);
    \draw (o) to[bend right=10] node[left=4pt, font=\small] {$H^2$} (u2);
    \draw (o) -- node[right=1pt, font=\small] {$H^3$} (u3);
    \draw (o) to[bend left=10] node[right=4pt, font=\small] {$H^4$} (u4);
    \draw (o) to[bend left=15] node[right=8pt, font=\small] {$H^5$} (u5);

    \foreach \t [evaluate=\t as \x using (\t-3)*2.8] in {1,...,5} {
        \foreach \i [evaluate=\i as \y using -\i*2] in {1,2,3} {
            \node[vertex] (w\t\i) at (\x, \y) {$w^{\t}_{\i}$};
            \draw (u\t) to[bend right=30] node[pos=0.8, right, yshift=1pt, font=\small] {$1$} (w\t\i);
        }
    }

    \foreach \t in {1,2,3,4} {
        \foreach \i in {1,2,3} {
            \draw (w\t\i) to
              node[pos=0.33, above, font=\small, yshift=1pt] {$\lceil V^\t_\i\rceil$}
              node[pos=0.33, below, font=\small, yshift=-1pt] {$\lfloor V^\t_\i\rfloor$}
              (w\the\numexpr\t+1\relax\i);
        }
    }

    \node[vertex] (d) at (9, -4) {$d$};
    \foreach \i in {1,2,3} {
        \draw (w5\i) to
            node[pos=0.3, above, sloped, font=\small] {$\lceil V^5_\i\rceil$}
            node[pos=0.3, below, sloped, font=\small] {$\lfloor V^5_\i\rfloor$}
            (d);
    }

\end{tikzpicture}

    \caption{
    Illustration of the flow network used in the proof of \Cref{prop:offline} for $T=5$ steps and $n=3$ parties.
    A single arc label represents an upper capacity; labels below and above an arc correspond to lower and upper capacities, respectively.
    }
    \label{fig:offlineflownetwork}
\end{figure}

%% file: splitterpath.tex
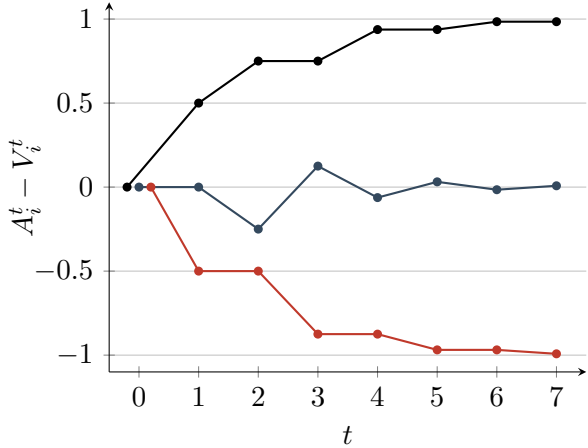
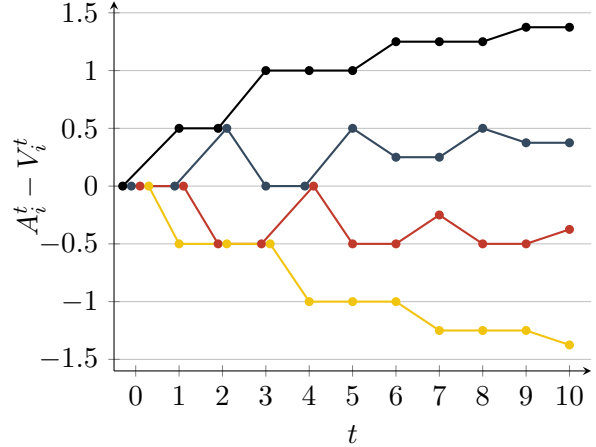
\begin{figure}[t]
\centering

\begin{subfigure}[t]{0.48\textwidth}
\centering
\begin{tikzpicture}
  \begin{axis}[
      width=\textwidth,
      height=6.5cm,
      xmin=-0.5, xmax=7.5,
      ymin=-1.1, ymax=1.1,
      axis lines=left,
      xtick={0,1,...,7},
      ytick={-1,-0.5,0,0.5,1},
      xlabel={$t$},
      ylabel={$A^t_i - V^t_i$},
      ylabel shift={-10pt},
      grid=major,
      major grid style={line width=.2pt,draw=gray!50},
      xmajorgrids=false,
      enlargelimits=false,
      clip=false,
  ]

  \addplot[black, thick, mark=*, mark size=1.3pt, mark options={fill=black}]
    coordinates {
      (-0.2, 0)
      (1, 0.5)
      (2, 0.75)
      (3, 0.75)
      (4, 0.9375)
      (5, 0.9375)
      (6, 0.984375)
      (7, 0.984375)
    };

  \addplot[blue, thick, mark=*, mark size=1.3pt, mark options={fill=blue}]
    coordinates {
      (0, 0)
      (1, 0)
      (2, -0.25)
      (3, 0.125)
      (4, -0.0625)
      (5, 0.03125)
      (6, -0.015625)
      (7, 0.0078125)
    };

  \addplot[red, thick, mark=*, mark size=1.3pt, mark options={fill=red}]
    coordinates {
      (0.2, 0)
      (1, -0.5)
      (2, -0.5)
      (3, -0.875)
      (4, -0.875)
      (5, -0.96875)
      (6, -0.96875)
      (7, -0.9921875)
    };

  \end{axis}
\end{tikzpicture}
\caption{For $n=3$, a surplus of $63/64$ and a deficit of $127/128$ are reached up to $t=7$.}
\end{subfigure}
\hfill
\begin{subfigure}[t]{0.48\textwidth}
\centering
\begin{tikzpicture}
  \begin{axis}[
      width=\textwidth,
      height=6.5cm,
      xmin=-0.5, xmax=10.5,
      ymin=-1.6, ymax=1.6,
      axis lines=left,
      xtick={0,1,...,10},
      ytick={-1.5,-1,-0.5,0,0.5,1,1.5},
      xlabel={$t$},
      ylabel={$A^t_i - V^t_i$},
      ylabel shift={-10pt},
      grid=major,
      major grid style={line width=.2pt,draw=gray!50},
      xmajorgrids=false,
      enlargelimits=false,
      clip=false,
  ]

  \addplot[black, thick, mark=*, mark size=1.3pt, mark options={fill=black}]
    coordinates {
      (-0.3, 0)
      (1, 0.5)
      (1.9, 0.5)
      (3, 1)
      (4, 1)
      (5, 1)
      (6, 1.25)
      (7, 1.25)
      (8, 1.25)
      (9, 1.375)
      (10, 1.375)
    };

  \addplot[blue, thick, mark=*, mark size=1.3pt, mark options={fill=blue}]
    coordinates {
      (-0.1, 0)
      (0.9, 0)
      (2.1, 0.5)
      (3, 0)
      (3.9, 0)
      (5, 0.5)
      (6, 0.25)
      (7, 0.25)
      (8, 0.5)
      (9, 0.375)
      (10, 0.375)
    };

  \addplot[red, thick, mark=*, mark size=1.3pt, mark options={fill=red}]
    coordinates {
      (0.1, 0)
      (1.1, 0)
      (1.9, -0.5)
      (2.9, -0.5)
      (4.1, 0)
      (5, -0.5)
      (6, -0.5)
      (7, -0.25)
      (8, -0.5)
      (9, -0.5)
      (10, -0.375)
    };

  \addplot[green, thick, mark=*, mark size=1.3pt, mark options={fill=green}]
    coordinates {
      (0.3, 0)
      (1, -0.5)
      (2.1, -0.5)
      (3.1, -0.5)
      (4, -1)
      (5, -1)
      (6, -1)
      (7, -1.25)
      (8, -1.25)
      (9, -1.25)
      (10, -1.375)
    };

  \end{axis}
\end{tikzpicture}
\caption{For $n=4$, a surplus and a deficit of $11/8$ are reached up to $t=10$.}
\end{subfigure}

\caption{Illustration of adversarially constructed partial instances to induce surpluses and deficits close to $(n-1)/2$, for $n\in \{3,4\}$. At each step $t\geq 1$, an $(i,j)$-splitter is applied for parties $i,j$ whose surpluses then split by one unit at the next step, while keeping the average constant. When multiple parties have the same surplus at some step, they are depicted with a small horizontal offset for ease of understanding.}
\label{fig:splitterpath}
\end{figure}

%% file: flownetwork.tex
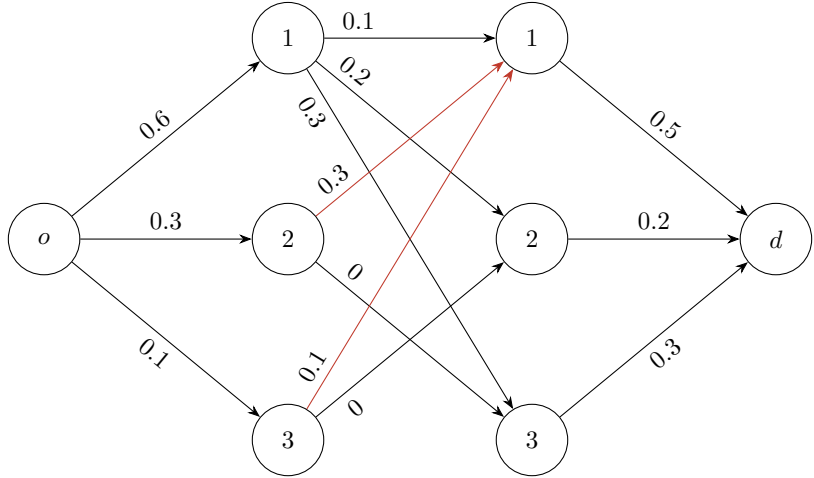
\begin{figure}[t]
    \centering
    \begin{tikzpicture}[
      node distance=1.8cm and 2.4cm,
      vertex/.style={draw, circle, minimum size=1cm, font=\small}, 
      ->, >=Stealth,
      scale=0.95, transform shape
    ]
    
        \node[vertex] (o) at (0,0) {$o$};
        \node[vertex] (a2) [right=of o] {$2$};
        \node[vertex] (a1) [above=of a2] {$1$};
        \node[vertex] (a3) [below=of a2] {$3$};
        \node[vertex] (b1) [right=of a1] {$1$};
        \node[vertex] (b2) [right=of a2] {$2$};
        \node[vertex] (b3) [right=of a3] {$3$};
        \node[vertex] (d) [right=of b2] {$d$};
        
        \draw (o) -- node[above, sloped, font=\small] {$0.6$} (a1);
        \draw (o) -- node[above, sloped, font=\small] {$0.3$} (a2);
        \draw (o) -- node[below, sloped, font=\small] {$0.1$} (a3);
        
        \draw (b1) -- node[above, sloped, font=\small] {$0.5$} (d);
        \draw (b2) -- node[above, sloped, font=\small] {$0.2$} (d);
        \draw (b3) -- node[below, sloped, font=\small] {$0.3$} (d);
        
        
        
        \draw (a1) -- node[pos=0.2, above, sloped, font=\small] {$0.1$} (b1);
        \draw (a1) -- node[pos=0.15, above, sloped, font=\small] {$0.2$} (b2);
        \draw (a1) -- node[pos=0.1, below, sloped, font=\small] {$0.3$} (b3);
        \draw[red] (a2) -- node[pos=0.15, above, sloped, font=\small] {\textcolor{black}{$0.3$}} (b1);
        \draw (a2) -- node[pos=0.15, above, sloped, font=\small] {$0$} (b3);
        \draw[red] (a3) -- node[pos=0.1, above, sloped, font=\small] {\textcolor{black}{$0.1$}} (b1);
        \draw (a3) -- node[pos=0.15, below, sloped, font=\small] {$0$} (b2);

    \end{tikzpicture}
    \caption{
    The flow network~$\normalfont\text{FN}(V^1,A^1,v^2,\pi)$ for~$v^1=(0.6,0.3,0.1)$ and~$v^2=(0.5,0.2,0.3)$.
    Edges with upper capacity zero are omitted. 
    The upper capacities of the remaining edges are~$c((o,\{i\}))=\pi(\{i\})=v_i^1$,~$c((\{i\},j))=\pi(\{i\})/H^1=v_i^1$, and~$c((i,d))=v_i^2$.
    An edge is highlighted in red if the lower capacity equals the upper capacity.
    The lower capacities of the remaining edges are zero.
    The labels indicate the unique flow of value~$1$.
    }
    \label{fig:flownetwork}
\end{figure}

%% file: fig_case1.tex
\begin{figure}[t]
    \centering
    \begin{subfigure}[b]{0.48\textwidth}
    \centering
    \begin{tikzpicture}[
      node distance=1.8cm and 2.4cm,
      vertex/.style={draw, circle, minimum size=1cm, font=\small}, 
      ->, >=Stealth,
      scale=0.65, transform shape
    ]
    
        \node[vertex] (o) at (0,0) {$o$};
        \node[vertex] (a2) [right=of o] {$2$};
        \node[vertex] (a1) [above=of a2] {$1$};
        \node[vertex] (a3) [below=of a2] {$3$};
        \node[vertex] (b1) [right=of a1] {$1$};
        \node[vertex] (b2) [right=of a2] {$2$};
        \node[vertex] (b3) [right=of a3] {$3$};
        \node[vertex] (d) [right=of b2] {$d$};
        
        \draw (o) -- node[above, sloped, font=\small] {$\pi(\{1\})$} (a1);
        \draw (o) -- node[above, sloped, font=\small] {$\pi(\{2\})$} (a2);
        \draw (o) -- node[below, sloped, font=\small] {$\pi(\{3\})$} (a3);
        
        \draw (b1) -- node[above, sloped, font=\small] {$v^{t+1}_1$} (d);
        \draw (b2) -- node[above, sloped, font=\small] {$v^{t+1}_2$} (d);
        \draw (b3) -- node[below, sloped, font=\small] {$v^{t+1}_3$} (d);
        
        \draw (a1) -- (b2);
        \draw (a1) -- (b3);
        \draw (a2) -- (b1);
        \draw (a2) -- (b3);
        \draw (a3) -- (b1);
        \draw (a3) -- (b2);
        
    \end{tikzpicture}
    \caption{Case 1.1.1}
    \label{fig:case1.1.1}
    \end{subfigure}
    \begin{subfigure}[b]{0.48\textwidth}
    \centering
    \begin{tikzpicture}[
      node distance=1.8cm and 2.4cm,
      vertex/.style={draw, circle, minimum size=1cm, font=\small}, 
      ->, >=Stealth,
      scale=0.65, transform shape
    ]
    
        \node[vertex] (o) at (0,0) {$o$};
        \node[vertex] (a2) [right=of o] {$2$};
        \node[vertex] (a1) [above=of a2] {$1$};
        \node[vertex] (a3) [below=of a2] {$3$};
        \node[vertex] (b1) [right=of a1] {$1$};
        \node[vertex] (b2) [right=of a2] {$2$};
        \node[vertex] (b3) [right=of a3] {$3$};
        \node[vertex] (d) [right=of b2] {$d$};
        
        \draw (o) -- node[above, sloped, font=\small] {$\pi(\{1\})$} (a1);
        \draw (o) -- node[above, sloped, font=\small] {$\pi(\{2\})$} (a2);
        \draw (o) -- node[below, sloped, font=\small] {$\pi(\{3\})$} (a3);
        
        \draw (b1) -- node[above, sloped, font=\small] {$v^{t+1}_1$} (d);
        \draw (b2) -- node[above, sloped, font=\small] {$v^{t+1}_2$} (d);
        \draw (b3) -- node[below, sloped, font=\small] {$v^{t+1}_3$} (d);
        
        \draw (a1) -- (b1);
        \draw (a1) -- (b2);
        \draw (a1) -- (b3);
        \draw[red] (a2) -- (b1);
        \draw[gray] (a2) -- (b3);
        \draw[red] (a3) -- (b1);
        \draw[gray] (a3) -- (b2);

    \end{tikzpicture}
    \caption{Case 1.1.2}
    \label{fig:case1.1.2}
    \end{subfigure}
    
    \vspace{.5cm}
    
    \begin{subfigure}[b]{0.48\textwidth}
    \centering
    \begin{tikzpicture}[
      node distance=1.8cm and 2.4cm,
      vertex/.style={draw, circle, minimum size=1cm, font=\small}, 
      ->, >=Stealth,
      scale=0.65, transform shape
    ]
    
        \node[vertex] (o) at (0,0) {$o$};
        \node[vertex] (a2) [right=of o] {$1,3$};
        \node[vertex] (a1) [above=of a2] {$1,2$};
        \node[vertex] (a3) [below=of a2] {$2,3$};
        \node[vertex] (b1) [right=of a1] {$1$};
        \node[vertex] (b2) [right=of a2] {$2$};
        \node[vertex] (b3) [right=of a3] {$3$};
        \node[vertex] (d) [right=of b2] {$d$};
        
        \draw (o) -- node[above, sloped, font=\small] {$\pi(\{1,2\})$} (a1);
        \draw (o) -- node[above, sloped, font=\small] {$\pi(\{1,3\})$} (a2);
        \draw (o) -- node[below, sloped, font=\small] {$\pi(\{2,3\})$} (a3);
        
        \draw (b1) -- node[above, sloped, font=\small] {$v^{t+1}_1$} (d);
        \draw (b2) -- node[above, sloped, font=\small] {$v^{t+1}_2$} (d);
        \draw (b3) -- node[below, sloped, font=\small] {$v^{t+1}_3$} (d);
        
        \draw (a1) -- (b1);
        \draw (a1) -- (b3);
        \draw (a2) -- (b1);
        \draw (a2) -- (b2);
        \draw[red] (a3) -- (b1);

    \end{tikzpicture}
    \caption{Case 1.2.1 }
    \label{fig:case1.2.1}
    \end{subfigure}
    \begin{subfigure}[b]{0.48\textwidth}
    \centering
    \begin{tikzpicture}[
      node distance=1.8cm and 2.4cm,
      vertex/.style={draw, circle, minimum size=1cm, font=\small}, 
      ->, >=Stealth,
      scale=0.65, transform shape
    ]
    
        \node[vertex] (o) at (0,0) {$o$};
        \node[vertex] (a2) [right=of o] {$1,3$};
        \node[vertex] (a1) [above=of a2] {$1,2$};
        \node[vertex] (a3) [below=of a2] {$2,3$};
        \node[vertex] (b1) [right=of a1] {$1$};
        \node[vertex] (b2) [right=of a2] {$2$};
        \node[vertex] (b3) [right=of a3] {$3$};
        \node[vertex] (d) [right=of b2] {$d$};
        
        \draw (o) -- node[above, sloped, font=\small] {$\pi(\{1,2\})$} (a1);
        \draw (o) -- node[above, sloped, font=\small] {$\pi(\{1,3\})$} (a2);
        \draw (o) -- node[below, sloped, font=\small] {$\pi(\{2,3\})$} (a3);
        
        \draw (b1) -- node[above, sloped, font=\small] {$v^{t+1}_1$} (d);
        \draw (b2) -- node[above, sloped, font=\small] {$v^{t+1}_2$} (d);
        \draw (b3) -- node[below, sloped, font=\small] {$v^{t+1}_3$} (d);
        
        \draw[gray] (a1) -- (b2);
        \draw[red] (a1) -- (b3);
        \draw[red] (a2) -- (b2);
        \draw[gray] (a2) -- (b3);
        \draw (a3) -- (b1);
        \draw (a3) -- (b2);
        \draw (a3) -- (b3);

    \end{tikzpicture}
    \caption{Case 1.2.2}
    \label{fig:case1.2.2}
    \end{subfigure}
    \caption{Case 1 of the proof of \Cref{lem:flowpolytope} with~$H^{t+1}=1$.
    The upper capacity of an edge~$e=(u,i)$ is~$c(e)=\pi(u)$ if it is drawn and~$c(e)=0$, otherwise.
    Red edges have lower capacities~$\ell(e)=c(e)$.
    Grey edges carry no flow in the unique flow with value~$1$.
    }
    \label{fig:case1}
\end{figure}
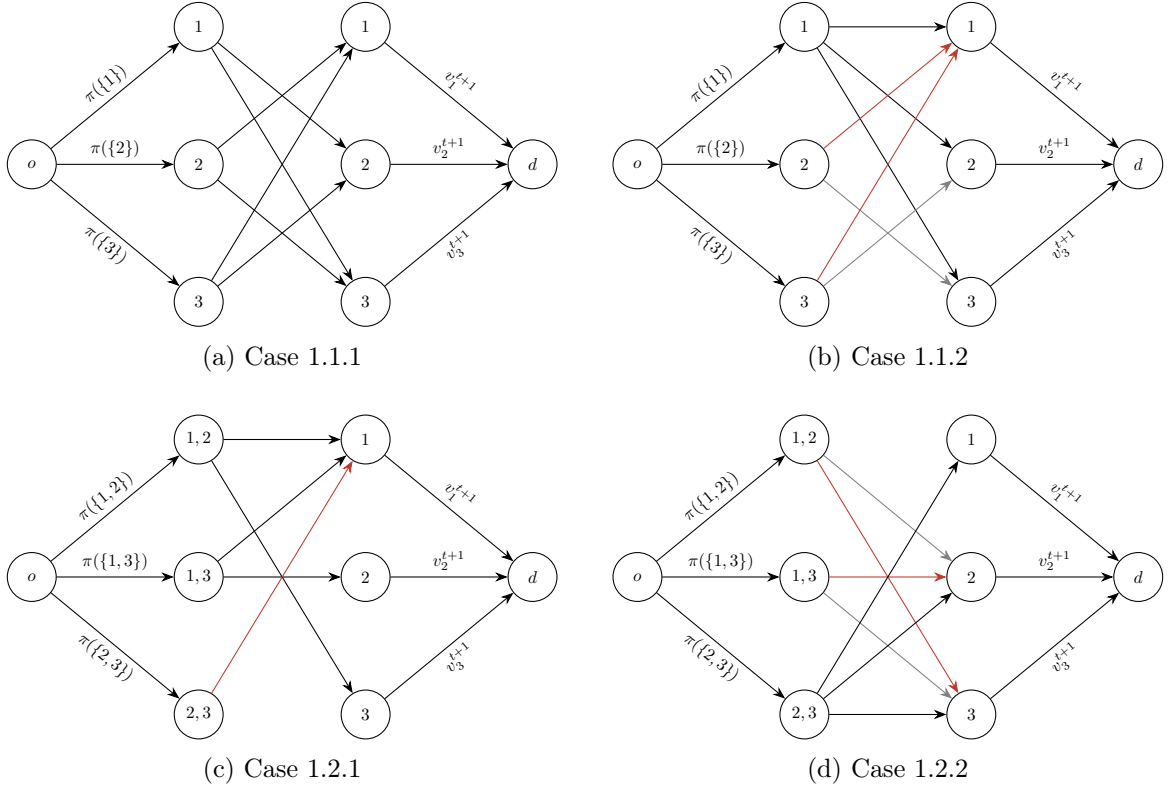

%% file: fig_case2.tex
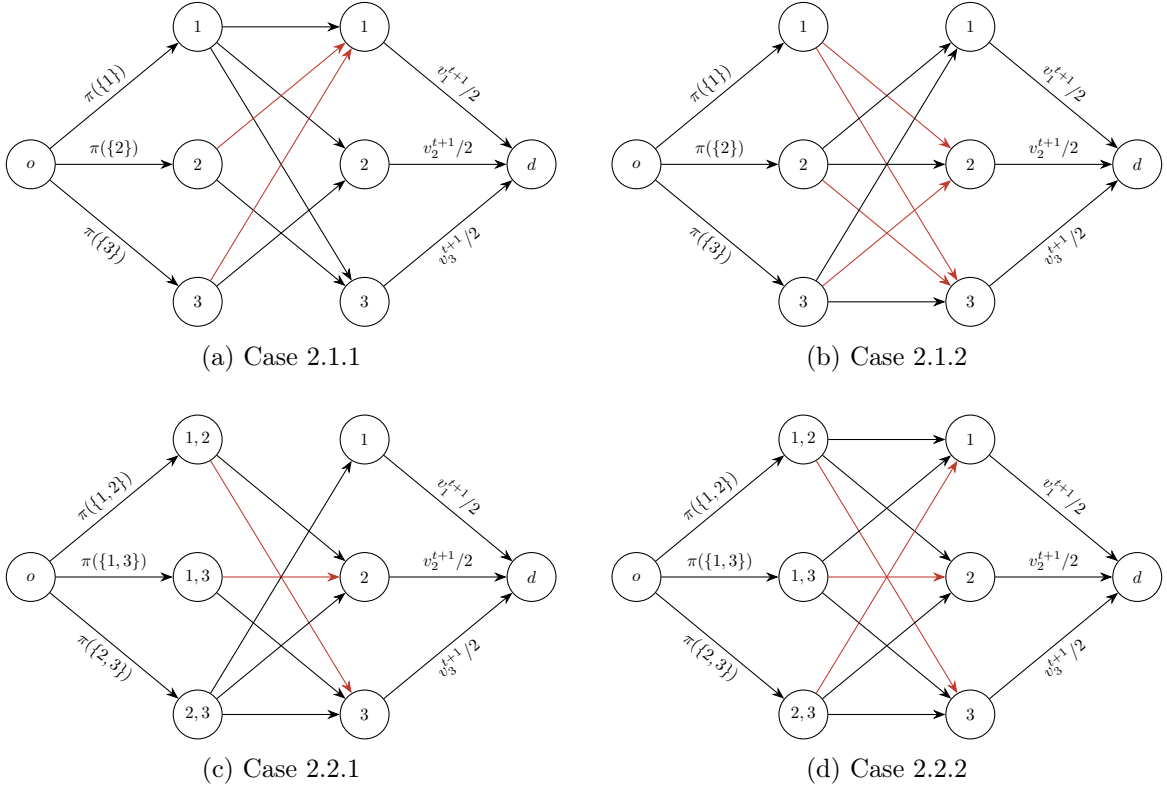
\begin{figure}[t]
    \centering
    \begin{subfigure}[b]{0.48\textwidth}
    \centering
    \begin{tikzpicture}[
      node distance=1.8cm and 2.4cm,
      vertex/.style={draw, circle, minimum size=1cm, font=\small}, 
      ->, >=Stealth,
      scale=0.65, transform shape
    ]
    
        \node[vertex] (o) at (0,0) {$o$};
        \node[vertex] (a2) [right=of o] {$2$};
        \node[vertex] (a1) [above=of a2] {$1$};
        \node[vertex] (a3) [below=of a2] {$3$};
        \node[vertex] (b1) [right=of a1] {$1$};
        \node[vertex] (b2) [right=of a2] {$2$};
        \node[vertex] (b3) [right=of a3] {$3$};
        \node[vertex] (d) [right=of b2] {$d$};
        
        \draw (o) -- node[above, sloped, font=\small] {$\pi(\{1\})$} (a1);
        \draw (o) -- node[above, sloped, font=\small] {$\pi(\{2\})$} (a2);
        \draw (o) -- node[below, sloped, font=\small] {$\pi(\{3\})$} (a3);
        
        \draw (b1) -- node[above, sloped, font=\small] {$v^{t+1}_1/2$} (d);
        \draw (b2) -- node[above, sloped, font=\small] {$v^{t+1}_2/2$} (d);
        \draw (b3) -- node[below, sloped, font=\small] {$v^{t+1}_3/2$} (d);
        
        \draw (a1) -- (b1);
        \draw (a1) -- (b2);
        \draw (a1) -- (b3);
        \draw[red] (a2) -- (b1);
        \draw (a2) -- (b3);
        \draw[red] (a3) -- (b1);
        \draw (a3) -- (b2);
        
    \end{tikzpicture}
    \caption{Case 2.1.1}
    \label{fig:case2.1.1}
    \end{subfigure}
    \begin{subfigure}[b]{0.48\textwidth}
    \centering
    \begin{tikzpicture}[
      node distance=1.8cm and 2.4cm,
      vertex/.style={draw, circle, minimum size=1cm, font=\small}, 
      ->, >=Stealth,
      scale=0.65, transform shape
    ]
    
        \node[vertex] (o) at (0,0) {$o$};
        \node[vertex] (a2) [right=of o] {$2$};
        \node[vertex] (a1) [above=of a2] {$1$};
        \node[vertex] (a3) [below=of a2] {$3$};
        \node[vertex] (b1) [right=of a1] {$1$};
        \node[vertex] (b2) [right=of a2] {$2$};
        \node[vertex] (b3) [right=of a3] {$3$};
        \node[vertex] (d) [right=of b2] {$d$};
        
        \draw (o) -- node[above, sloped, font=\small] {$\pi(\{1\})$} (a1);
        \draw (o) -- node[above, sloped, font=\small] {$\pi(\{2\})$} (a2);
        \draw (o) -- node[below, sloped, font=\small] {$\pi(\{3\})$} (a3);
        
        \draw (b1) -- node[above, sloped, font=\small] {$v^{t+1}_1/2$} (d);
        \draw (b2) -- node[above, sloped, font=\small] {$v^{t+1}_2/2$} (d);
        \draw (b3) -- node[below, sloped, font=\small] {$v^{t+1}_3/2$} (d);
        
        \draw[red] (a1) -- (b2);
        \draw[red] (a1) -- (b3);
        \draw (a2) -- (b1);
        \draw (a2) -- (b2);
        \draw[red] (a2) -- (b3);
        \draw (a3) -- (b1);
        \draw[red] (a3) -- (b2);
        \draw (a3) -- (b3);

    \end{tikzpicture}
    \caption{Case 2.1.2}
    \label{fig:case2.1.2}
    \end{subfigure}
    
    \vspace{.5cm}
    
    \begin{subfigure}[b]{0.48\textwidth}
    \centering
    \begin{tikzpicture}[
      node distance=1.8cm and 2.4cm,
      vertex/.style={draw, circle, minimum size=1cm, font=\small}, 
      ->, >=Stealth,
      scale=0.65, transform shape
    ]
    
        \node[vertex] (o) at (0,0) {$o$};
        \node[vertex] (a2) [right=of o] {$1,3$};
        \node[vertex] (a1) [above=of a2] {$1,2$};
        \node[vertex] (a3) [below=of a2] {$2,3$};
        \node[vertex] (b1) [right=of a1] {$1$};
        \node[vertex] (b2) [right=of a2] {$2$};
        \node[vertex] (b3) [right=of a3] {$3$};
        \node[vertex] (d) [right=of b2] {$d$};
        
        \draw (o) -- node[above, sloped, font=\small] {$\pi(\{1,2\})$} (a1);
        \draw (o) -- node[above, sloped, font=\small] {$\pi(\{1,3\})$} (a2);
        \draw (o) -- node[below, sloped, font=\small] {$\pi(\{2,3\})$} (a3);
        
        \draw (b1) -- node[above, sloped, font=\small] {$v^{t+1}_1/2$} (d);
        \draw (b2) -- node[above, sloped, font=\small] {$v^{t+1}_2/2$} (d);
        \draw (b3) -- node[below, sloped, font=\small] {$v^{t+1}_3/2$} (d);
        
        \draw (a1) -- (b2);
        \draw[red] (a1) -- (b3);
        \draw[red] (a2) -- (b2);
        \draw (a2) -- (b3);
        \draw (a3) -- (b1);
        \draw (a3) -- (b2);
        \draw (a3) -- (b3);

    \end{tikzpicture}
    \caption{Case 2.2.1}
    \label{fig:case2.2.1}
    \end{subfigure}
    \begin{subfigure}[b]{0.48\textwidth}
    \centering
    \begin{tikzpicture}[
      node distance=1.8cm and 2.4cm,
      vertex/.style={draw, circle, minimum size=1cm, font=\small}, 
      ->, >=Stealth,
      scale=0.65, transform shape
    ]
    
        \node[vertex] (o) at (0,0) {$o$};
        \node[vertex] (a2) [right=of o] {$1,3$};
        \node[vertex] (a1) [above=of a2] {$1,2$};
        \node[vertex] (a3) [below=of a2] {$2,3$};
        \node[vertex] (b1) [right=of a1] {$1$};
        \node[vertex] (b2) [right=of a2] {$2$};
        \node[vertex] (b3) [right=of a3] {$3$};
        \node[vertex] (d) [right=of b2] {$d$};
        
        \draw (o) -- node[above, sloped, font=\small] {$\pi(\{1,2\})$} (a1);
        \draw (o) -- node[above, sloped, font=\small] {$\pi(\{1,3\})$} (a2);
        \draw (o) -- node[below, sloped, font=\small] {$\pi(\{2,3\})$} (a3);
        
        \draw (b1) -- node[above, sloped, font=\small] {$v^{t+1}_1/2$} (d);
        \draw (b2) -- node[above, sloped, font=\small] {$v^{t+1}_2/2$} (d);
        \draw (b3) -- node[below, sloped, font=\small] {$v^{t+1}_3/2$} (d);
        
        \draw (a1) -- (b1);
        \draw (a1) -- (b2);
        \draw[red] (a1) -- (b3);
        \draw (a2) -- (b1);
        \draw[red] (a2) -- (b2);
        \draw (a2) -- (b3);
        \draw[red] (a3) -- (b1);
        \draw (a3) -- (b2);
        \draw (a3) -- (b3);

    \end{tikzpicture}
    \caption{Case 2.2.2}
    \label{fig:case2.2.2}
    \end{subfigure}
    \caption{Case 2 of the proof of \Cref{lem:flowpolytope} with~$H^{t+1}=2$.
    The upper capacity of an edge~$e=(u,i)$ is~$c(e)=\pi(u)/2$ if it is drawn and~$c(e)=0$, otherwise.
    Red edges have lower capacities~$\ell(e)=c(e)$.
    }
    \label{fig:case2}
\end{figure}

%% file: fig_case3.tex
\begin{figure}[t]
    \centering
    \begin{tikzpicture}[
      node distance=1.8cm and 2.4cm,
      vertex/.style={draw, circle, minimum size=1cm, font=\small}, 
      ->, >=Stealth,
      scale=0.65, transform shape
    ]
    
        \node[vertex] (o) at (0,0) {$o$};
        \node[vertex] (a2) [right=of o] {$1,3$};
        \node[vertex] (a1) [above=of a2] {$1,2$};
        \node[vertex] (a3) [below=of a2] {$2,3$};
        \node[vertex] (b1) [right=of a1] {$1$};
        \node[vertex] (b2) [right=of a2] {$2$};
        \node[vertex] (b3) [right=of a3] {$3$};
        \node[] (test) [right=of b2] {};
        \node[vertex] (d) [right=of test] {$d$};
        
        \draw (o) -- node[above, sloped, font=\small] {$\pi(\{1,2\})/2$} (a1);
        \draw (o) -- node[above, sloped, font=\small] {$\pi(\{1,3\})/2$} (a2);
        \draw (o) -- node[below, sloped, font=\small] {$\pi(\{2,3\})/2$} (a3);
        
        \draw (b1) -- node[above, sloped, font=\small] {$(v^{t+1}_1-\pi(\{2,3\}))/2$} (d);
        \draw (b2) -- node[above, sloped, font=\small] {$(v^{t+1}_2-\pi(\{1,3\}))/2$} (d);
        \draw (b3) -- node[below, sloped, font=\small] {$(v^{t+1}_3-\pi(\{1,2\}))/2$} (d);
        
        \draw (a1) -- (b1);
        \draw (a1) -- (b2);
        \draw (a2) -- (b1);
        \draw (a2) -- (b3);
        \draw (a3) -- (b2);
        \draw (a3) -- (b3);

    \end{tikzpicture}
    \caption{The residual flow network of Case 2.2.2 in the proof of \Cref{lem:flowpolytope}.
    The upper capacity of an edge~$e=(u,i)$ is~$c(e)=\pi(u)/2$ if it is drawn and~$c(e)=0$, otherwise.
    }
    \label{fig:case3}
\end{figure}
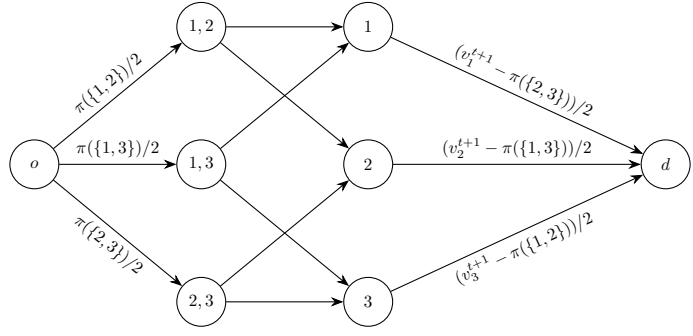